\documentclass[10pt,letter]{IEEEtran}
\usepackage{amsmath}
\usepackage{amsthm}
\usepackage{amsfonts}
\usepackage{amssymb}
\usepackage{graphicx}
\usepackage{epstopdf}
\usepackage{url,ifthen}
\usepackage{caption}
\usepackage{cite}
\usepackage{subcaption}
\usepackage{mathrsfs}
\usepackage{nopageno}
\usepackage[dvipsnames,usenames]{color}
\usepackage[lmargin=0.76in,rmargin=0.76in,tmargin=0.73in,bmargin=0.79in]{geometry}

\newtheorem{definition}{Definition}
\newtheorem{lem}{Lemma}
\newtheorem{prop}{Proposition}
\newtheorem{theorem}{Theorem}
\newtheorem{comment}{Comment}

\def\eqdef{\ensuremath{:=}}
\def\predHorzLen{\ensuremath{N_p^-}}

\newcommand{\version}{arxiv}
\newboolean{showArxiv}
\newboolean{showArxivAlt}

\ifthenelse{\equal{\version}{journal}}
{
	\setboolean{showArxiv}{false} 
	\setboolean{showArxivAlt}{true}
}
{
}

\ifthenelse{\equal{\version}{arxiv}}
{
	\setboolean{showArxiv}{true} 
	\setboolean{showArxivAlt}{false}
}
{
}

\ifshowArxivAlt
\graphicspath{{figures/}}
\fi

\begin{document}
\bstctlcite{IEEEexample:BSTcontrol} %needed so repeated authors aren't ---
\title{\vspace{6.4mm} Predictive resource allocation for flexible loads with local QoS}
\author{
	\IEEEauthorblockN{Austin R. Coffman\IEEEauthorrefmark{1}$^,$\IEEEauthorrefmark{2}, Matthew Hale\IEEEauthorrefmark{1}, and Prabir Barooah\IEEEauthorrefmark{1}}  \vspace{-0.5cm}
	\thanks{\IEEEauthorrefmark{2} corresponding author, email: bubbaroney@ufl.edu.}
	\thanks{\IEEEauthorrefmark{1} AC, MH, and PB are with the Dept. of Mechanical and Aerospace Engineering, University of Florida, Gainesville, FL 32601, USA. AC and PB are partially supported by the NSF through award 1646229 (CPS-ECCS). MH was supported by ONR under grant N00014-19-1-2543 and by a task order from the Munitions Directorate of AFRL.}
	
}
\maketitle
\thispagestyle{empty}
\begin{abstract}
Loads that can vary their power consumption without violating their Quality of service (QoS), that is flexible loads, are an invaluable resource for grid operators. Utilizing flexible loads as a resource requires the grid operator to incorporate them into a resource allocation problem. Since flexible loads are often consumers, for concerns of privacy it is desirable for this problem to have a distributed implementation. Technically, this distributed implementation manifests itself as a time varying convex optimization problem constrained by the QoS of each load. In the literature, a time invariant form of this problem without all of the necessary QoS metrics for the flexible loads is often considered. Moving to a more realistic setup introduces additional technical challenges, due to the problems' time-varying nature. In this work, we develop an algorithm to account for the challenges introduced when considering a time varying setup with appropriate QoS metrics.

%In a centralized setting, for a grid operator to incorporate flexible loads into a resource allocation problem, the \emph{capacity} of the collection is required. To know this, the grid operator requires extensive knowledge about the loads and the resulting capacity estimate is highly dependent on the grid operators information being correct. This process is thus highly invasive and non-robust modeling errors.

%Rather, we focus on a distributed solution, where only each load needs to know its own QoS parameters and coordination is achieved through communication. This distributed framework is considered often in the literature. However, often the loads QoS are not fully modeled; static and box constrains are only considered. Incorporating the constraints that better represent the QoS of the loads requires    
\end{abstract}
\section{Introduction}

Relying more and more on renewable generation is the envisioned future for the power grid. However, this goal is not without its challenges; renewable sources, such as solar and wind, are highly volatile. Moreover, supply and demand of power must always be in equilibrium, and when renewable generation cannot ensure this, controllable generation sources must ramp to ensure equilibrium. Economically, for a Balancing Authority (BA) (the institution responsible for ensuring supply and demand are balanced in a given geographical area), ramping generators or utilizing batteries for this is not feasible. This has motivated the recent investigation of a new resource to help where conventional generators and batteries fall short: flexible loads. 

%An envisioned future for the power grid is one that relies more on renewable generation sources. An inevitable challenge in this scenario is the inherent variability present in renewable generation sources, such as solar or wind. This variability requires grid operators to ramp controllable resources up and down to meet the demand when renewable generation does not. Ramp rate constraints prevent conventional generation from handling this mismatch completely. Grid level storage from batteries is expensive. Thus a new resource is being investigated to help fill the mismatch where conventional generators and batteries fall short: flexible loads.

Flexible loads can deviate from a baseline level of consumption without violating the Quality of Service (QoS) of the load. From the perspective of the BA, flexible loads deviating from baseline are identical to a battery discharging and charging. Due to this, flexible loads are often said to provide ``Virtual Energy Storage'' (VES)~\cite{bar:2019}. More importantly, grid support from flexible loads is more cost effective than batteries~\cite{cammardella2018energy}. Some examples of flexible loads include residential air conditioners~\cite{CoffmanVESBuildSys:2018}, water heaters~\cite{liu2019trajectory}, refrigerators~\cite{mathias2016smart}, commercial HVAC systems~\cite{haokowlinbarmey:2013}, and pumps for irrigation~\cite{AghajFarm:2019} and pool cleaning~\cite{chenDistributedIMA:2017}.

%Flexible loads have the ability to vary power consumption over a baseline level without violating the Quality of Service (QoS) of the load. The baseline power consumption is the power consumed without grid interference. The requested amount, from the grid authority, to deviate from baseline is termed the \emph{reference signal}. The tracking of the reference signal guises, in the eyes of the grid operator, flexible loads as batteries providing storage services. 

To utilize flexible loads, the BA in some way must incorporate them into a resource allocation problem. In a centralized framework, the resource allocation problem involves a central authority accounting for all of its resources and their constraints, and then allocating its needs to each resource based on the constraints. The problem is typically solved for a specific future duration. For instance, the BA allocates its resources for the next day~\cite{cammerdellaCDC:2018}.  

In contrast to the BA solving a centralized resource allocation problem, it is possible to decentralize and have each flexible load solve a portion of the centralized problem. Furthermore, this distributed algorithm can run in real time. The advantage of a distributed solution is that (i) privacy is protected, as each load only needs to know its own QoS and (ii) the solution is more robust to modeling error as no one entity is making decisions for the ensemble based on models of the ensemble; each member of the ensemble makes decisions for itself based on a combination of its local and global information. 

Solving the resource allocation problem in a distributed fashion and real time falls under the framework of \emph{time-varying optimization}. There are two main challenges in this framework: (C1) shifting to a real time solution is problematic for constraints with ``memory'', e.g. dynamic systems or rate constraints that require past state values to evaluate, and (C2) at each instant in time typically only one iteration of the optimization algorithm can be applied. While the effects of point (C2) are indirectly/directly analyzed in virtually all works on real time optimization, point (C1) is often not considered. That is, most works on time varying optimization focus only on static constraint maps~\cite{TAC_Sun:2017} or unconstrained problems~\cite{simonetto2016class}. Unfortunately, the QoS of flexible loads is specified by constraints with memory.

In addition to the literature on \emph{time-varying optimization}, there is a subfield of literature focused on the distributed resource allocation for flexible loads~\cite{BrooksDecentralizedTSG:2019,brobar:ACC:2016,Hug_distOpt:2015TSG,zhatoplilow:2014,CherCortes:2018TAC,baiSunFengHu:CDC2018,zhao2013optimal,doan2017distributed} in the smart grid. While there is a library of work~\cite{coffman2019aggregate:arxivACC,CoffmanCharacterizingArxiv:2020,hao_aggregate:2015} on how to model the QoS of flexible loads for the purpose of resource allocation, only a few works on distributed resource allocation take this into account~\cite{BrooksDecentralizedTSG:2019,CherCortes:2018TAC}.

To summarize, much of the past work on time-varying optimization is focused on problems of different structure than the resource allocation problem for flexible loads. Thus the algorithms developed are not directly applicable. Whereas, many of the past works focused on distributed resource allocation for flexible loads do not account for the entirety of the loads' QoS.
%Consequently, the capacity is a function of every loads QoS and the solution of the resource allocation problem is highly dependent on the BA having accurate knowledge of each loads QoS. Furthermore, having the grid operator perform a centralized resource allocation method is non-robust to modeling errors.

%One key challenge for the BA is determining how the QoS of each individual flexible load limits the ability of the ensemble to track a reference signal. Or inversely, how a BA can allocate a reference signal so that the ensemble can track this signal without any individual violating their QoS. Many works are dedicated to this~\cite{}, and the set of constraints on the reference signal induced from the individuals QoS is termed the ``capacity'' of the collection. The determination of the capacity is inherently centralized, and requires the BA to have knowledge of each loads ``QoS parameters.'' Thus any reference signal planned with the capacity is highly dependent upon the ``QoS parameters'' being correct. 

%Instead of a centralized solution approach for resource allocation, we focus on a distributed approach. That is, instead of the central authority allocating its needs to each load, each load determines how it can contribute to the grids needs.  

In this paper, we develop an algorithm for distributed resource allocation that allows loads to account for a wide variety of QoS metrics. In doing so, our algorithm incorporates principled techniques to overcome the challenges (C1) and (C2) listed above. To overcome (C1), we employ a \emph{state augmentation} technique that augments past fictitious state values that act as surrogates for the previous states. To overcome (C2), we utilize predictions of the time varying quantities to facilitate a benefit similar to warm start technique in centralized optimization. With all features of the algorithm accounted for, we prove an Input to State stability (ISS) result for when the time varying aspect is arbitrary (but in some sense bounded). This stability result is guaranteed under gain conditions that are specified in terms of the readily available problem data.

In numerical experiments we validate our theoretical results and compare our proposed method to a past method in the literature. In the time-varying setting our proposed method is able to successfully have flexible loads solve the resource allocation problem in a distributed/hierarchical fashion. Additionally, it is shown that the past method, based on dual ascent, can lead to integrator windup in the same time-varying setting.    

The paper proceeds as follows: in Section~\ref{sec:gridNeed} the problem setup and requirements are described. In Section~\ref{sec:probSetup} the resource allocation problem is introduced, as well as past ways it has been posed as an optimization problem. In Section~\ref{sec:propMethod} our proposed method is introduced and it is analyzed in Section~\ref{sec:stability}. We give numerical examples in Section~\ref{sec:numExp} and conclude in Section~\ref{sec:conc}.
  
\section{Needs of the loads and the power grid} \label{sec:gridNeed}

\subsection{Notation}
We let $\mathbb{N}$ and $\mathbb{R}$ denote the natural and real numbers, respectively. We let the index $i \in\{1,\dots,N\}$ denote the $i^{th}$ load, where $N$ is the total number of loads. The index $t \in \mathbb{N}$ is the discrete time index. The index $t$ will only appear as a subscript, while $i$ will only appear as a superscript. In the sequel, unless specified $\|\cdot\|$ will refer to the $2$-norm of a vector on the appropriate dimension vector space. We reserve lowercase letters for vectors/scalars and uppercase letters for matrices. The notation $x[j]$, when $x$ is a vector, will refer to the $j^{th}$ element of the vector $x$.

The power consumed by load $i$ at time $t$ is denoted $d^i_{t|t}$. Furthermore, the quantity $d^i_{t+j|t}$ is the power consumption that at time $t$ load $i$ predicts it will consume at time $t+j$, where $j\leq N_p$ and $N_p$ is the prediction horizon. For convenience, we define $\predHorzLen \eqdef N_p - 1$. The required total power from all loads, i.e., the reference signal, at time $t$ is denoted $s_t$. 

We consider two ``stacked'' vectorized versions of the scalar quantities $d^i_{t+j|t}$. The first is the \emph{load perspective stacking} where we stack the scalars $d^i_{t+j|t}$ into a vector and denote it as $x^i_{t}\triangleq [d^i_{t|t},\dots,d^i_{t+\predHorzLen|t}]^T$. The second is the \emph{grid perspective stacking} where we stack over all loads, forming $x_{j|t} \triangleq [d^1_{t+j|t},\dots,d^N_{t+j|t}]^T$. In any case, for a fixed $N_p$ we refer to to the following $x_t \triangleq [(x^1_{t})^T,\dots,(x^N_{t})^T]^T$, which contains all the elements of $x^i_{t}$ and $x_{j|t}$. The purpose for introducing both stacked forms is for ease of exposition. 

%At time $t$, each load will have access to the prediction error signal (to be described shortly) over the entire prediction horizon so to make a decision about its demand deviation, also at time $t$. At the next time step, $t+1$, the desired total power deviation changes to $S_{t+1}$. 

\begin{figure}
	\centering
	\includegraphics[width=0.9\columnwidth]{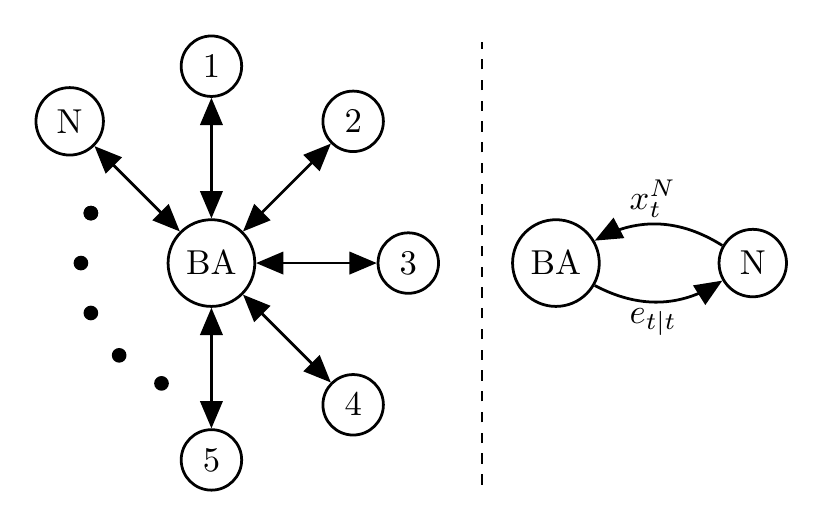}
	\caption{The information structure considered, which is representative of the structure of a utility (BA) in the USA. The numbers represent the flexible load index.}
	\label{fig:infoStruct}
\end{figure}
\subsection{BA's Needs: Reference tracking (global goal)}
The BA employs support from flexible loads to help mitigate supply and demand mismatch. Using the previously defined variables, this goal is captured by requiring the following to be small:
\begin{align} \label{eq:trackError}
&e_{\tau|t} = \sum_{i=1}^{N}d^i_{\tau|t}-s_\tau, \quad J_G(x_{\tau|t}) = e^2_{\tau|t},  \ \tau \leq t+ \predHorzLen.
\end{align}

%As seen in Figure~\ref{fig:tempDecLoad}, a single load will utilize the string of values, communicated through the grid authority, $[e_{t|t},\dots,e_{t+N_p|t}]$ to figure out how it will contribute to $S_t$.

\subsection{Individual Needs: The QoS set (local constraints)} \label{sec:QoSSet}
We describe the requirements of the loads through a QoS set. These constraints are taken from the vast literature on ``capacity characterization'' of flexible loads~\cite{coffman2019aggregate:arxivACC,CoffmanCharacterizingArxiv:2020,hao_aggregate:2015}. The constraints on the power for the $i^{th}$ load, $x^i_t$, are:
\begin{align} \label{eq:consSet}
&\mathcal{D}^i(d^i_{t-1|t-1}) \triangleq \\ \nonumber
&\bigg\{x^i_{t}: \forall \ j  \in \{t,\dots, t+\predHorzLen\}, \\  \label{eq:powerCon} 
&\textbf{Power}: \quad d^i_L \leq d^i_{j|t} \leq d^i_H, \\ \label{eq:rateCon} 
&\textbf{Rate}: \quad r^i_L \leq d^i_{j|t} - d^i_{j-1|t} \leq r^i_H, \ j > t\\ \label{eq:rateIC}
&\textbf{Rate-IC}: \quad r^i_L \leq d^i_{t|t} - d^i_{t-1|t-1} \leq r^i_H,
\\ \label{eq:engCon}
&\textbf{Energy}: \quad e^i_L \leq \sum_{j=t}^{t+\predHorzLen}d^i_{j|t} \leq e^i_H \bigg\}.
\end{align}  
Each constraint~\eqref{eq:powerCon}-\eqref{eq:engCon} has a specific meaning as illustrated by the labels given. An additional Rate-IC constraint is included to emphasize that previous data is required to evaluate this constraint. Furthermore, it is necessary to define the QoS set~\eqref{eq:consSet} over a time horizon, otherwise enforcing the constraint~\eqref{eq:engCon} would not be possible. The constraints~\eqref{eq:powerCon}-\eqref{eq:engCon} model various classes of flexible loads, e.g. batteries, HVAC systems in commercial buildings, thermostatically controlled loads (TCLs)~\cite{CoffmanCharacterizingArxiv:2020}, and pool pumps~\cite{cammerdellaCDC:2018}. 

However, while the QoS set specifies maximum limits it does not mean that it is desirable to operate at these limits. Thus, the loads are also interested in making the following quantity small,
\begin{align} \label{eq:regTerm}
&J_{L}(x_{\tau|t}) = \sum_{i=1}^{N}(d^i_{\tau|t})^2\zeta^i, \quad \tau \leq t+ \predHorzLen,
\end{align}
where $\zeta^i > 0$ for all $i \in \{1,\dots,N\}$. The quantity~\eqref{eq:regTerm} can be thought of as a regularization term.
\begin{prop}
	For each $i$ the set $\mathcal{D}^i(d^i_{t-1|t-1})$ is compact, convex, and non empty.
\end{prop}

There are two important points about the set $\mathcal{D}^i(d^i_{t-1|t-1})$: (i) the constraints~\eqref{eq:rateCon}-\eqref{eq:engCon} require more than one instant of time to appropriately evaluate, and (ii) the constraint~\eqref{eq:rateIC} has memory, at time $t$ the set $\mathcal{D}^i(d^i_{t-1|t-1})$ is a function of $d^i_{t-1|t-1}$.   
  
\begin{comment}
The constraint set~\eqref{eq:consSet} in its abstract form captures the heterogeneity of the load. In fact, other than it being convex, compact, and each load having an independent constraint set, we require no more assumptions for this set. For example, load $i=10$ could be a Walmart and load $i=5$ could be a classroom on a university campus, both shifting their load to help the grid. Put explicitly, our proposed optimization problem and solution method tolerate arbitrary high degrees of heterogeneity.
\end{comment}
  
\subsection{Information structure}
The information structure considered is depicted in Figure~\ref{fig:infoStruct}, which is a hierarchical communication structure with distributed computation. For each time $t$, the loads are allowed to communicate exactly \emph{once} to the BA in order to receive global information, the signal $e_{t|t}$~\eqref{eq:trackError}. The loads can then use this global information to apply \emph{one} iteration of an optimization algorithm to achieve the global goal, tracking the reference $s_t$. However, at the next time, $t+1$, the reference $s_t$ will change and hence the optimization problem the loads are attempting to solve is operating in ``real time.''  
  
\section{Resource allocation and optimization basics} \label{sec:probSetup}

%The QoS methodology in section~\ref{sec:indQoS} culminated in a discrete time constraint on the PSD of each loads power consumption. This constraint will be utilized in the following, and as such the analysis will solely take place in discrete time.
%The resource allocation problem is a way to jointly specify all the goals described in Section~\ref{sec:gridNeed}. Additionally, due to the structure of the resource allocation problem, it can readily be solved in a distributed fashion. 

The goal of the resource allocation problem is to set up one problem that combines both the grid's and individual needs, as specified in Section~\ref{sec:gridNeed}. Additionally, we seek a distributed and real time solution to the resource allocation problem. As stated in the introduction, the combination of the requirements in Section~\ref{sec:gridNeed} with a real time and distributed implementation is often not considered. 

To better understand our contribution we review resource allocation problems considered in past literature, and comment on how these methods lead to challenges when faced with the more realistic problem specifications here. However, before any of this we review how to solve a constrained optimization problem, of special structure, in a distributed fashion using projected gradient descent. 

%To better understand our contribution, we (i) describe how to solve a constrained optimization problem, of special structure, in a distributed fashion, and (ii) list some typical resource allocation problems that have been considered in the past literature. 
\subsection{Solving a constrained optimization problem}
A distributed algorithm for solving the following time varying structured convex problem,
\begin{align} \label{prob:refOptProb}
	\min_{z\in\mathcal{Z}}\ f(z;t), \quad z \in \mathbb{R}^q, \quad \mathcal{Z} = \mathcal{Z}^1 \times \dots \times \mathcal{Z}^q,
\end{align} 
with $z^i \in \mathcal{Z}^i$ only,  is the so-called projected gradient descent method,
 \begin{align} \label{alg:expAlg}
	 z^i_{t+1} &= \Pi_{\mathcal{Z}^i}\bigg(z^i_{t}-\alpha\nabla f^i(z_t)\bigg), \quad \forall i \in \{1,\dots,q\}, \\ \nonumber
	 \nabla f^i(z_t) &\triangleq \frac{\partial f(z;t)}{\partial z^i}\bigg\vert_{z = z_t}, \quad \Pi_{\mathcal{X}}(x) \triangleq \arg\min_{y\in\mathcal{X}}\|y-x\|,
 \end{align}
with $\alpha>0$ a step size. The projected gradient method applied to time invariant problems has its origins in~\cite{goldstein1964}. For an introduction to time varying convex optimization the paper~\cite{Popkov2005} is a good reference. As we will see, the resource allocation problem naturally has a similar structure to~\eqref{prob:refOptProb}. 

\ifx 0
\subsection{Example Resource allocation 1}
One way that the resource allocation problem can be posed is through the following optimization problem,
\begin{align} \label{prob:compProb}
\nu^*_t  = &\min_{D_{t}}\ \frac{1}{2}\bigg(J_1(D_{t|t}) + J_2(D_{t|t}) \bigg) \\ \label{prob:compProbCon1}
&\text{s.t.} \quad D^i_{|t} \in [d^i_L,d^i_H], \quad \forall i \in \{1,\dots,N\}, 
\end{align}
where the prediction horizon $N_p = 0$. We see the structure of problem~\eqref{prob:compProb} exactly fits the structure of~\eqref{prob:refOptProb} so we can solve~\eqref{prob:compProb} with
\begin{align} \label{eq:exPrimSol}
D^i_{|t+1} = \Pi_{[d^i_L,d^i_H]}\bigg(D^i_{|t} - \alpha\nabla\nu(D^i_{|t})\bigg), \ \forall i \in\{1,\dots,N\},
\end{align} 
where $\alpha$ is the stepsize. However, having the loads implement~\eqref{eq:exPrimSol} is likely to lead to QoS violation, since many of the QoS metrics are not considered in this method.
\fi
\subsection{Example resource allocation 1: Dual ascent} \label{sec:dualMethod}
A commonly encountered resource allocation problem~\cite{zhatoplow:2012} is,
\begin{align} \label{prob:compProb2}
\beta^*_t  = &\min_{x_t} \ \beta(x_t) = \frac{1}{2}\bigg(J_{L}(x_{t|t}) + J_G(x_{t|t})\bigg) \\ 
&\text{s.t.} \quad d^i_{t|t} \in [d^i_L,d^i_H], \quad \forall i \in\{1,\dots,N\}, \\ \label{prob:compProb2End}
&e_{t|t} = 0 \ \leftrightarrow \ \lambda_t,
\end{align}
where $\leftrightarrow$ refers to the association of the dual variable $\lambda_t$ (i.e., the Lagrange multiplier). In this setting the prediction horizon, $N_p$, is zero making the decision variable for load $i$ only $d^i_{t|t}$. Given feasibility and strong convexity, one can solve this problem in a distributed/hierarchical fashion with the so called ``dual ascent'' method,
\begin{align} 
\label{eq:unstableInt}
\lambda_{t} &= \lambda_{t-1} + \gamma e_{t-1|t-1}, \\
d^i_{t|t} &= \Pi_{[d^i_L,d^I_H]}\big(\frac{\lambda_t}{\zeta^i}\big),
\end{align}
with $\gamma$ a stepsize. The general derivation of these equations can be found in most introductory optimization textbooks~\cite{bertsekas1989parallel}. 

Immediately, we see that this method will have some problems. Firstly, the resource allocation~\eqref{prob:compProb2}-\eqref{prob:compProb2End} do not account for all of the constraints in the QoS set. Secondly, if $e_{t|t}$ cannot be made small (ideally zero) then the solution method~\eqref{eq:unstableInt} will suffer from the so called ``integrator windup'' phenomenon. For a time-invariant optimization problem, under the appropriate assumptions, it is straightforward to ensure zero steady state error, i.e., $e_{t|t}\rightarrow 0$. However, when the optimization problem is non-stationary it may be possible that for some time the problem is feasible and for other periods of time the problem is not feasible. When the problem is non-feasible the dual update equation~\eqref{eq:unstableInt} will continue to integrate non-zero error. When the problem becomes feasible again, the Lagrange multiplier will be far away from the optimal Lagrange multiplier for the newly feasible problem. Thirdly, knowing that the resource allocation~\eqref{prob:compProb2} will be feasible for all time is centralized knowledge, or requires the solution to a centralized optimization problem. 

\subsection{Example resource allocation 2: predictive resource allocation} \label{sec:predResAlloc}
Another resource allocation formulation is the predictive resource allocation problem, which is described by the following optimization problem at time $t\in\mathbb{N}$:
\begin{align}  
\label{prob:predResAlloc}
\kappa^*_t = &\min_{x_t} \ \kappa(x_t) = \frac{1}{2}\bigg(\sum_{\tau=t}^{t+\predHorzLen} J_{L}(x_{\tau|t})+J_G(x_{\tau|t})\bigg) \\ \label{eq:propBoundPO}
&\text{s.t.} \quad x^i_{t} \in \mathcal{D}^i(d^i_{t-1|t-1}), \quad \forall \ i\in\{1,\dots,N\},
\end{align}
with $\predHorzLen > 0$.
This formulation allows for the incorporation of an appropriate QoS set~\eqref{eq:consSet}, however the constraint set is time varying and state dependent. So, while this problem may appear to be in the form amendable for the algorithm~\eqref{alg:expAlg}, this is not the case. As the algorithm~\eqref{alg:expAlg} requires a fixed constraint set, and the constraint set~\eqref{eq:propBoundPO} is not fixed. So as it stands, there is no clear way to specify a distributed algorithm to solve~\eqref{prob:predResAlloc}.

This problem is considered in~\cite{Hug_distOpt:2015TSG}, however the focus there is not a real time implementation. As a result, the challenges we face here were not present in~\cite{Hug_distOpt:2015TSG}.

\emph{The formulation of a resource allocation problem with all the appropriate QoS constraints, such as \eqref{prob:predResAlloc}-\eqref{eq:propBoundPO} -  and an algorithm for its solution are the focus of the rest of the paper.}
\section{Proposed Method} \label{sec:propMethod}
Largely, the limitation of the past resource allocation problem is that they do not consider appropriate load QoS metrics (dual ascent resource allocation, Section~\ref{sec:dualMethod}). Further, we see that when including the appropriate metrics the constraint set becomes time varying and state dependent (predictive resource allocation, Section~\ref{sec:predResAlloc}). So that if we wish to use the appropriate QoS set, modifications to the resource allocation must be done to make the set fixed. We handle this limitation with a \emph{state augmentation} technique, which we describe next.

\subsection{Predictive Resource Allocation with memory}
We define the \emph{memory} objective at time $t\in\mathbb{N}$ as follows: $J_{\text{m}}(x_{t-1|t})\triangleq$
\begin{align}
	  \sum_{i=1}^{N}(d^i_{t-1|t} - d^i_{t-1|t-1})^2\bar{\zeta}^i + \bigg(\sum_{i=1}^{N}d^i_{t-1|t} - s_{t-1}\bigg)^2.
\end{align}
We have introduced the variable $d^i_{t-1|t}$,  which is a fictitious variable at time $t$ that we desire to be close to $d^i_{t-1|t-1}$ (treated as a constant at time $t$), where close is defined by $(d^i_{t-1|t} - d^i_{t-1|t-1})^2$. The augmented decision variable, $z_t$, containing $d^i_{t-1|t}$ is then:
\begin{align}
	z^i_{t} &\triangleq [d^i_{t-1|t}, (x^i_{t})^T]^T, \\
	z_t &\triangleq [(z^1_{t})^T,\dots,(z^N_{t})^T]^T,
\end{align}
where, by construction, $z_t$ contains all the elements in  $x_{j|t}$, so where convenient we refer to $x_{j|t}$ however, within the scope of an optimization problem, the relevant decision variable is $z_t$. With $z^i_{t}$ it is now possible to redefine the QoS set~\eqref{eq:consSet} as independent of the previous state value. We denote this new set as:
\begin{align} \label{eq:timeInvSet}
\mathcal{D}^i \triangleq &\bigg\{z^i_{t}: %\forall \ j  \in \{t-1,\dots, t+N_p\}, \\ 
\ \text{s.t.}~\eqref{eq:powerCon}, \eqref{eq:rateCon}, \text{and} \  \eqref{eq:engCon}\bigg\}.
\end{align}

\begin{comment}
	In~\eqref{eq:timeInvSet} the constraint~\eqref{eq:rateIC} is evaluated with the decision variable $d^i_{t-1|t}$ and not an externally specified variable/parameter. Hence, there is no need to distinguish between the rate and rate-IC constraint.
\end{comment}
With this, the predictive resource allocation problem with memory is the following: 
\begin{align}
\nonumber
&\min_{z_t} \ \eta(z_t)=\frac{1}{2}\bigg(\sum_{\tau=t}^{t+\predHorzLen} J_{L}(x_{\tau|t})+J_{G}(x_{\tau|t}) + J_{\text{m}}(x_{t-1|t})\bigg)\\ \label{prob:predMemResAlloc}
&\text{s.t.} \quad z^i_{t} \in \mathcal{D}^i, \quad \forall \ i\in\{1,\dots,N\}.
\end{align}
We see that~\eqref{prob:predMemResAlloc} is in a form applicable to the example algorithm~\eqref{alg:expAlg}. The solution to~\eqref{prob:predMemResAlloc} is denoted $z^*_t$ with optimal value $\eta^*_t = \eta(z^*_t)$. 
\ifx 0
\section{DRA: Dual algorithm}

Consider the DRA optimization problem at time $t \in \mathbb{N}$ specified by~\eqref{prob:dualOnly}. We will solve this problem with a dual method, which as we shall see is distributed in nature. A weakness of this approach is that feasibility of~\eqref{prob:dualOnly} is required for the convergence results presented.

\subsection{Distributed Dual Algorithm}
Since the problem~\eqref{prob:dualOnly} is convex and satisfies Slater's constraint qualification, there is zero duality gap. That is, the optimal value of the dual problem and the primal problem are the same. The dual problem is characterized by the dual function
\begin{align} \label{eq:dualFunc}
	&\phi^*(\lambda) = \min_{D^i\in\mathcal{D}^i} \mathcal{L}(D,\lambda) = \sum_{i=1}^{N}\frac{\zeta^i}{2}\left(D^i\right)^TD^i + \lambda(S_t-D^i_t),
\end{align}
where $\mathcal{L}(D,\lambda)$ is the Lagrangian. 
The dual function $\phi^*(\lambda)$ is strictly concave, and the dual problem is
\begin{align}\label{eq:dualProb}
	\lambda^* = \arg\max_{\lambda} \ \phi^*(\lambda).
\end{align}
The primal variable is recovered from $\lambda^*$ by solving for all $i\in\{1,\dots,N\}$,
\begin{align} \label{eq:primSubProb}
	D^{i,*}(\lambda^*) = \arg\min_{D^i \in\mathcal{D}^i} \frac{\zeta^i}{2}\left(D^i\right)^TD^i - \lambda^*D^i_t.
\end{align}
To obtain $\lambda^*$ we solve~\eqref{eq:dualProb} with the so called ``dual ascent''~\cite{} method:
\begin{align} \label{eq:dualIter}
	\lambda_{t+1} = \lambda_t + \gamma \nabla\phi^*(\lambda_t).
\end{align}
The iteration~\eqref{eq:dualIter} is distributed as,
\begin{align} \nonumber
\nabla\phi^*(\lambda_t) &= S_t - \sum_{i=1}^{N}D^{i,*}(\lambda_t),
\end{align}
The gradient can be computed locally, see~\cite{zhatoplow:2013} for further discussions on this. Practically, in applications where local computation is not possible the above quantity can be shared through an undirected star graph, and is thus still distributed. Furthermore, the primal sub problem~\eqref{eq:primSubProb} separates over $i$, and can thus be solved by each agent individually, given it had access to $\lambda_t$.

We now present a lemma that will be useful in developing stability results for the distributed resource allocation problem.

\begin{lem} \label{lem:lipCont}
	The function $\nabla\phi^*(\lambda)$ is Lipschitz on $\mathbb{R}$ with constant $L=\sum_{i=1}^{N}\frac{1}{\zeta^i}$.
\end{lem} 

\begin{proof}
	We proceed in a direct manner,
	\begin{align} \nonumber
		\left|\nabla\phi^*(\lambda_1)-\nabla\phi^*(\lambda_2)\right| &= \bigg|\sum_{i=1}^{N}D^{i,*}(\lambda_1)-D^{i,*}(\lambda_2)\bigg| \\ \nonumber
		&\leq \sum_{i=1}^{N}\bigg|D^{i,*}(\lambda_1)-D^{i,*}(\lambda_2)\bigg|,
	\end{align}
	where the bound is by the triangle inequality. Now, $D^{i,*}(\lambda_1) = \Pi_{\mathcal{D}^i}(\frac{\lambda_1}{\zeta^i})$ and the projection operator, $\Pi_{\mathcal{X}}(x)$, is non-expansive~\cite{} so,
	\begin{align} \nonumber
	\sum_{i=1}^{N}\bigg|D^{i,*}(\lambda_1)-D^{i,*}(\lambda_2)\bigg| &\leq \big|\lambda_1-\lambda_2\big|\sum_{i=1}^{N}\frac{1}{\zeta^i}. 	
	\end{align}
	We have dropped the absolute value on $|1/\zeta^i|$, as this quantity is always positive.
	This completes the proof by setting, as desired, $L = \sum_{i=1}^{N}\frac{1}{\zeta^i}$.
\end{proof}

For concreteness, we re-write the distributed dual algorithm together,
\begin{align} \label{eq:dualAlg}
	\lambda_{t} &= \lambda_{t-1} + \gamma\bigg(S_{t-1} - \sum_{i=1}^{N}D^{i,*}_{t-1}\bigg), \\
	D^{i,*}_{t} &= \arg\min_{D^i \in\mathcal{D}^i} \frac{\zeta^i}{2}\left(D^i\right)^TD^i - \lambda_tD^i_{t-1}.
\end{align}

A next key step in the analysis is to show what we might expect the tracking error to be. Or in the language of optimization theory, the constraint violation of constraint~\eqref{eq:propSDMatchDO}.

\subsection{Stability Results}
We provide two stability results, based on the type of signal that $S_t$ is.

We first present the result for a step reference signal, then provide a result for a time varying reference signal. The step reference signal is a well known result in the literaure. Whereas the time varying reference signal is a new result. For the time varying reference signal, the main idea is to show that the constraint violation is uniformly ultimately bounded (UUB), and that this bound provides reasonable tracking error.

\begin{theorem}
	If the gain $\gamma$ satisfies $\gamma \in \big(0,\frac{1}{L}\big]$ and $S_t$ is a step then,
	\begin{align} \nonumber
	\lim_{t\rightarrow\infty} \left|\lambda_t-\lambda^*\right| = 0,
	\end{align}
	if $\{S_t, \ t\in\mathbb{N}\}$ is a time varying signal then,
	\begin{align} \nonumber
		\bigg|\sum_{i=1}^{N}D^{i,*}(\lambda_t) - S_{t}\bigg| < B(\lambda_0), \quad  \forall \ t \in \mathbb{N},
	\end{align}
	where $\lambda_0$ is the initial iterate of~\eqref{eq:dualAlg}.
\end{theorem}
\begin{proof}
	
	At the optimal solution, i.e. $\lambda^*$, the constraint violation is zero, so we can freely add it to the non-optimal constraint violation,
	\begin{align} \nonumber
	\bigg|\sum_{i=1}^{N}D^{i,*}(\lambda_t) - S_{t}\bigg| =& \bigg|\sum_{i=1}^{N}\bigg(D^{i,*}(\lambda_t)-D^{i,*}(\lambda^*_t)\bigg) \bigg| \\ \nonumber
	\leq& L\left|\lambda_t - \lambda^*_t\right|,
	\end{align} 
	where the bound is from the Lipschitz constant obtained from Lemma~\ref{lem:lipCont}. We then have
	\begin{align} \label{eq:consViolFinal}
	\bigg|\sum_{i=1}^{N}D^{i,*}(\lambda_t) - S_{t}\bigg| \leq \bigg|\lambda_{t} - \lambda_t^*\bigg|\sum_{i=1}^{N}\frac{1}{\zeta^i}
	\end{align}
	In any case, we have that $\sum_{i=1}^{N}\frac{1}{\zeta^i} < \frac{1}{\gamma}$, so the time varying bound on the constraint violation is,
	\begin{align}
	\bigg|\sum_{i=1}^{N}D^{i,*}(\lambda_t) - S_{t}\bigg| < \frac{\left|\lambda_{t} - \lambda_t^*\right|}{\gamma}
	\end{align}
	We rewrite $\left|\lambda_{t} - \lambda_t^*\right|$ as,
	\begin{align}
	\left|\lambda_{t-1} - \lambda_t^*\right| &= \left|\lambda_{t-1} - \lambda_t^*+\lambda_{t-1}^*-\lambda_{t-1}^*\right| \\
	&\leq \left|\lambda_{t-1} - \lambda_{t-1}^*\right| + \left|\lambda^*_{t-1} - \lambda_t^*\right| \\
	&\leq \gamma\left|\lambda_{t-2} - \lambda_{t-1}^*\right| + \left|\lambda^*_{t-1} - \lambda_t^*\right|
	\end{align}
	proceeding back to $t=0$, we obtain
	\begin{align}
	\left|\lambda_{t} - \lambda_t^*\right| \leq \gamma^{t}\left|\lambda_{0} - \lambda_0^*\right| + \sum_{i=1}^{t}\gamma^{t-i}\left|\lambda^*_{i-1} - \lambda_{i}^*\right|
	\end{align}
	Now, we can take the limit $t\rightarrow\infty$ so to obtain a uniform bound, provided for all $t\in\mathbb{N}$ we have that $\left|\lambda^*_{t} - \lambda_{t-1}^*\right| < \Lambda < \infty$,
	\begin{align}
	\forall \ t\in\mathbb{N}, \quad \left|\lambda_{t} - \lambda_t^*\right| < \frac{\Lambda}{1-\gamma} + \left|\lambda_{0} - \lambda_0^*\right|
	\end{align}
	So the constraint violation is then bounded as,
	\begin{align} \nonumber
	\bigg|\sum_{i=1}^{N}D^{i,*}(\lambda_t) - S_{t}\bigg| < B(\lambda_0) = \bigg(\frac{\Lambda}{1-\gamma} + \left|\lambda_{0} - \lambda_0^*\right|\bigg)
	\end{align}
	and is thus UUB. When the reference is constant, the UUB term results to asymptotic stability.
	
	\ifx 0
	which is a countable set of finite valued points, now define
	\begin{align}
	\overline{\Lambda} \triangleq \big[\inf(\Lambda),\ \sup(\Lambda)\big],
	\end{align}
	clearly $\lambda_t \in \overline{\Lambda}$ for all $t \in \mathbb{N}$. So that the final bound is,
	\begin{align}
	\bigg|\sum_{i=1}^{N}D^{i,*}(\lambda_t) - S_{t}\bigg| < B = \frac{\alpha\Lambda^\bullet}{\gamma} , \ \text{where} \\
	\Lambda^\bullet = \sup(\Lambda) - \inf(\Lambda)
	\end{align} 
	and thus the constraint violation is UUB. 
	\fi
\end{proof}

\fi
\subsection{Proposed algorithm}

%We now pose an algorithm that can solve the predictive resource allocation problem with memory~\eqref{prob:predMemResAlloc}. The algorithm will be a modified version of~\eqref{alg:expAlg} to better account for the additional structure of the problem. 

\ifx 0
\subsection{Preliminaries}
The time varying nature of the problem makes it hard to establish formal guarantees about how well the algorithm will work. The most challenging aspect is perhaps the time varying constraint set $\mathcal{D}^i_k$. Unfortunately, this change is not arbitrary but rather state dependent, i.e., the values of the previous demand consumptions for load $i$ will affect the constraints in $\mathcal{D}^i_k$. The reason why this is a problem is in how one can recursively define the optimal trajectory for load i, $D^{i,*}_t$. The \emph{appropriate} definition is as follows
\begin{align}
	D^*_t = \Pi_{\mathcal{D}^*}\bigg(D^*_t - \alpha\nabla\eta(D^*_t)\bigg),
\end{align}
where $\mathcal{D}^*$ is the constraint set computed with the previous optimal solutions. However, it should be noted that the above does not hold if the set $\mathcal{D}^*$ is not computed with the previous optimal values, i.e., for some set $\mathcal{X} \neq \mathcal{D}^*$
\begin{align}
	D^*_t \neq \Pi_{\mathcal{X}}\bigg(D^*_t - \alpha\nabla\eta(D^*_t)\bigg),
\end{align}
which is troublesome as the typical solution method,
\begin{align}
	D_t = \Pi_{\mathcal{D}}\bigg(D_{t-1} - \alpha\nabla\eta(D_{t-1})\bigg) 
\end{align}
no longer shares the same projection set, so that the standard non-expansive property is not directly applicable, 
\begin{align} \nonumber
	\left|\left|D_t - D^*_t\right|\right| = \left|\left|\Pi_{\mathcal{D}}(\psi_t) - \Pi_{\mathcal{D}^*}(\psi_t^*)\right|\right| \nleq \left|\left|\psi_t - \psi^*_t\right|\right|.
\end{align}
In this work we overcome this through the process of \emph{state-augmentation}. More precisely, the initial conditions that make the constraint set $\mathcal{D}$ state dependent are augmented into the state of the system. These values are then ``projected'' onto the dummy variable defined in the projection operator. Thus offering a uniform constraint set over time, and between algorithm and the optimal solution.
\fi

To solve the problem~\eqref{prob:predMemResAlloc}, we propose the following algorithm. The $i^{th}$ load updates its state with:
\begin{align}
	z^i_{t+1} %&=%\Pi_{\mathcal{D}^i}\bigg(\big(1-\alpha\zeta^i\big)D^i_{t} - \alpha\bigg(\sum_{i=1}^{N}D^i_t - S_t\bigg)\bigg), \\ 
	\label{eq:distPrimAlg}
	&= \Pi_{\mathcal{D}^i}\bigg(\hat{P}\big(z^i_{t} - \alpha\nabla\eta^i(z_{t})\big)\bigg) = \Pi_{\mathcal{D}^i}\bigg(\hat{P}\psi^i_t\bigg), \\ \nonumber
	\psi^i_t &\triangleq z^i_{t} - \alpha\nabla\eta^i(z_{t}),
\end{align}
where $\alpha>0$ is a step size common to all loads,
and $\hat{P}$ is the following matrix,
\begin{align}
	\hat{P} = \begin{bmatrix}
	\mathbf{0}_{N_p\times 1}& I_{N_p} \\
	1& \mathbf{0}_{1\times N_p} \\
	\end{bmatrix}.
\end{align}
Including the matrix $\hat{P}$ will be elaborated on in section~\ref{sec:contrib}. However, its primary purpose is to ``shift'' the data to facilitate a benefit similar to warm start techniques in optimization. In fact, a flavor of this idea was included in~\cite{wang2009fast}, among others, to speed up the solution time for real time Model Predictive Control.

Recall, for each load $i$, the quantity $z^i_{t}$ is a vector in $\mathbb{R}^{N_p+1}$ where $N_p$ is the prediction horizon. The algorithm~\eqref{eq:distPrimAlg} is an update rule for the entire vector $z^i_{t}$, the value that the load $i$ \emph{actually} consumes at time $t$ is then $z^i_{t}[2] = d^i_{t|t}$. 

The ensemble dynamics, i.e., the vectorized form of the algorithm~\eqref{eq:distPrimAlg} are,
\begin{align} \label{eq:fullAlg}
	&z_{t+1}= \big[\Pi_{\mathcal{D}^1}\big(\hat{P}\psi^1_t\big), \dots, \Pi_{\mathcal{D}^N}\big(\hat{P}\psi^N_t\big)\big]^T = \Pi_{\mathcal{D}}\bigg(P\psi_t\bigg), \\
	&\mathcal{D} = \mathcal{D}^1\times\dots\times\mathcal{D}^N,  \ \psi_t = [(\psi^1_t)^T,\dots,(\psi^N_t)^T]^T, \\
	&P=I_N\otimes\hat{P},
\end{align}
where $\times$ denotes Cartesian product, $\otimes$ denotes matrix Kronecker product~\cite{horn2012matrix}, and $z_{t}$ is a vector in $\mathbb{R}^{(N_p+1)N}$. Since the Cartesian product operation preserves convexity and for each $i$ we have $\mathcal{D}^i$ is convex, the set $\mathcal{D}$ is also convex. The vectorized form~\eqref{eq:fullAlg} is useful for analysis, however during implementation each load has the ability to  update its own local variable $z^i_{t}$ by solely using~\eqref{eq:distPrimAlg}.

\begin{prop} \label{prop:fixPoint}
Let $z^*_t$ be the optimal solution to  problem~\eqref{prob:predMemResAlloc} at time $t\in\mathbb{N}$, then we have that
	\begin{align} \nonumber
	z^*_t = \Pi_{\mathcal{D}}\bigg(z^*_t - \alpha\nabla\eta(z^*_t)\bigg) &= \Pi_{\mathcal{D}}\big(\psi^*_t\big). 
	\end{align}	
	%where in the gradient we use the following $d^i_{t-1|t} = d^{i,*}_{t-1|t-1}$, i.e., the previous optimal value.
\end{prop}
In the above, the set $\mathcal{D}$ is the \emph{same} set that the algorithm~\eqref{eq:fullAlg} uses during implementation. 
We will see that this facet of Proposition~\ref{prop:fixPoint} is important for the stability analysis of the proposed algorithm~\eqref{eq:fullAlg}.

\subsection{Contribution} \label{sec:contrib}
Our proposed resource allocation method and algorithm have three key contributions over the past literature: (i) we accurately account for all the QoS of the loads, (ii) the inclusion of predictions and consequently the ``horizon shifting'' matrix $\hat{P}$, and (iii) the inclusion of the term:
\begin{align} \label{eq:prevPenTerm}
J_{\text{m}}(x_{t-1|t})
\end{align}  
in the objective of~\eqref{prob:predMemResAlloc}. These improvements have been stated prior, but now with the developed math and notation they can be better exposed. The advantages of point (i) are explicit, so we focus on points (ii) and (iii). 

Elaborating on point (ii), 
multiplying the content of the projection operator by the matrix $\hat{P}$ in~\eqref{eq:distPrimAlg} is consistent with ``shifting the horizon'' of data. For instance, at time $t$ the $i^{th}$ load produces a trajectory of demand consumptions from time $t$ to time $t+\predHorzLen$. The initial condition at time $t+1$ is then the value predicted for time $t+1$ at time $t$. However, this value has already gone through at least one iteration thus speeding up the convergence in a way similar to ``warm start'' techniques in centralized optimization.

Elaborating on point (iii), we now examine what would happen when~\eqref{eq:prevPenTerm} is not included in the objective. We are concerned with the predictive resource allocation problem, as described in Section~\ref{sec:predResAlloc}. In this scenario the appropriate fixed point definition for the optimal trajectory in terms of the algorithm is now
\begin{align} \label{eq:predOptSol}
x^*_t &= \Pi_{\mathcal{D}^*_t}\bigg(x^*_t - \alpha\nabla\kappa(x^*_t)\bigg) = \Pi_{\mathcal{D}^*_t}\bigg(\Psi^*_t\bigg), \\ 
\mathcal{D}^*_t &\triangleq \mathcal{D}^1(d^{1,*}_{t-1|t-1})\times\dots\times\mathcal{D}^N(d^{N,*}_{t-1|t-1}),
\end{align}
thus $\mathcal{D}^*_t$ is the constraint set computed with the previous optimal values, and $\Psi^*_t$ is the content of the projection operator in~\eqref{eq:predOptSol}.
Suppose now we attempt to use the prototype algorithm~\eqref{alg:expAlg} to solve~\eqref{prob:predResAlloc}, i.e.,
\begin{align} \nonumber
	x_{t+1} &= \Pi_{\mathcal{D}_{t}}\bigg(x_{t} - \alpha\nabla\kappa(x_{t})\bigg) = \Pi_{\mathcal{D}_{t}}(\Psi_t), \\ \nonumber
	\mathcal{D}_t &\triangleq \mathcal{D}^1(d^{1}_{t-1|t-1})\times\dots\times\mathcal{D}^N(d^{N}_{t-1|t-1}).
\end{align} 
Typically, one is interested in bounding $\|x_t-x^*_t\|$, which is usually performed using the non-expansive property of the projection operator. However since we have that $\mathcal{D}_t \neq \mathcal{D}^*_t$, we see that,
\begin{align} \nonumber
\|x_t - x^*_t\| = \|\Pi_{\mathcal{D}_t}(\Psi_t) - \Pi_{\mathcal{D}^*_t}(\Psi_t^*)\| \nleq \|\Psi_t - \Psi^*_t\|,
\end{align}
since the non-expansive property requires that $\mathcal{D}_t = \mathcal{D}^*_t$. Thus showing convergence or boundedness will be greatly complicated, and likely lead to a lackluster bound.

\ifx 0
Our proposed algorithm~\eqref{eq:fullAlg} is inherently different from most of the optimization based algorithms for distributed control in the literature. These algorithms directly take the previous iterate and map it to the initial condition for the next iterate. If predictions of the signal $S_t$ are available, not shifting the horizon is disregarding important information that would otherwise allow the algorithms to converge faster. In fact, there is work on online algorithms to speed up model predictive control (MPC) using, among other things, a similar idea to this~\cite{wang2009fast}.  In MPC converging faster allows one to operate at faster time scales. In this distributed control context, converging faster is equivalent to reduced tracking error even in the presence of poor information. 

We further make clear the contribution, as the primal algorithm~\eqref{eq:distPrimAlg} used to solve~\eqref{prob:predMemResAlloc} may appear ubiquitous. However, this is \emph{not} the case as we have two key contributions: (i) the inclusion of predictions and consequently the ``horizon shifting'' matrix $\hat{P}$ and (ii) the inclusion of the term:
  
in the objective of~\eqref{prob:predMemResAlloc}, which can be thought of as \emph{state-augmentation} where we have augmented a fictitious state $d^i_{t-1|t}$ that we require to be close to the actual previous value $\hat{d}^i_{t-1|t-1}$. 

Since our constraints (and most constraints in engineering systems) have a ``memory'' aspect to them, without  including~\eqref{eq:prevPenTerm} in the objective would cause the constraint set~\eqref{eq:consSet} to be time varying. Instead, we have moved this time varying aspect to the objective function through the state augmentation process. The importance of a fixed constraint set is primarily for analysis; without this fixed constraint set, providing reasonable theoretical guarantees about the algorithm~\eqref{eq:fullAlg} is near impossible. We motivate this point through the following proposition, which we stress is possible due to the additional term~\eqref{eq:prevPenTerm} in~\eqref{prob:predMemResAlloc}.

In the above, the set $\mathcal{D}$ is the \emph{same} set that the algorithm~\eqref{eq:fullAlg} uses during implementation. 
\subsubsection{Hypothetical Scenario}
We now examine what would happen when~\eqref{eq:prevPenTerm} is not included in the objective. Thus the decision variable is $D_t$, as defined in Section~\ref{sec:gridNeed}. In this scenario the constraint set is time varying and state dependent, as the initial conditions for the constraints depend on the previous solution values. Furthermore, the appropriate fixed point definition for $D^*_t$ is now,
\begin{align}
D^*_t &= \Pi_{\mathcal{D}^*}\bigg(D^*_t - \alpha\nabla\eta(D^*_t)\bigg), \\ 
\mathcal{D}^* &\triangleq \mathcal{D}^1(d^{1,*}_{t-1|t-1})\times\dots\times\mathcal{D}^N(d^{N,*}_{t-1|t-1}),
\end{align}
thus $\mathcal{D}^*$ is the constraint set computed with the previous optimal values.
Since the algorithm~\eqref{eq:fullAlg} no longer shares the same projection set as the optimal trajectory, applying the standard non-expansive property of the projection is not directly applicable, for instance
\begin{align} \label{eq:failedProj}
\left|\left|D_t - D^*_t\right|\right| = \left|\left|\Pi_{\mathcal{D}}(\psi_t) - \Pi_{\mathcal{D}^*}(\psi_t^*)\right|\right| \nleq \left|\left|\psi_t - \psi^*_t\right|\right|,
\end{align}
since the non-expansive property requires that $\mathcal{D} = \mathcal{D}^*$.
This is problematic, as~\eqref{eq:failedProj} is the starting point for showing convergence to the optimal solution $D^*_t$. Thus showing convergence or boundedness will be greatly complicated, and likely lead to a lackluster bound.
\fi
\section{Stability} \label{sec:stability}
\subsection{Preliminaries}
\ifx0
We start with the following assumptions and definitions,
\begin{align}
	\left|\left|S_t\right|\right|_{\infty} &\triangleq \sup_{t\in\mathbb{N}}\left|S_t\right| < \infty, \\
	\bar{S}_\alpha 
	 &\triangleq \sup_{t\in\mathbb{N}}\left|S_{t}-S_{t+\alpha}\right| < \infty, \ \forall \ \alpha \in \mathbb{N}, \\ 
	\bar{S} &= \sup_{\alpha\in\mathbb{N}}\bar{S}_{\alpha} < \infty,
\end{align}
which in all practicality will never be violated. 
\fi
We list a string of results that will be useful for the analysis of the proposed algorithm~\eqref{eq:fullAlg}. %For the interested reader, these propositions (as well as the ones in previous sections) appear as lemmas with proofs in an expanded version of this paper~\cite{}.

\ifx 0
\begin{prop}
	The matrices $\hat{P}$ and $P$ have the following properties,
	\begin{align} \nonumber
		&\text{(i):} \quad \big\|\hat{P}(x-y)\big\| = \|x-y\|, \\ \nonumber &\text{(ii):} \quad\|P(x-y)\| = \|x-y\|, \\ \nonumber
		&\text{(iii):} \quad \hat{P}^{-1}  = \hat{P}^{T}, \ \text{and} \ P^{-1}  = I_{N}\otimes\hat{P}^{T},%\\ \nonumber 
		%&\text{(iv):}\quad P^N = I_N, \quad \text{for} \ P \in \mathbb{R}^{N\times N}
	\end{align}
\end{prop}
\fi
\begin{prop} \label{prop:hessStruc}
	The Hessian $\nabla^2\eta$ and gradient $\nabla\eta(z_t)$ can be expressed in the following form, letting $H^i \triangleq \text{diag}([\bar{\zeta}^i,\zeta^i,\dots,\zeta^i]) \in \mathbb{R}^{N_p+1}$, for all $z_t \in \mathbb{R}^{(N_p+1)N}$
	\begin{align} \nonumber
	&\text{(i)}: \quad \nabla^2\eta = \mathbf{1}_N\otimes\bigg( \mathbf{1}^T_N\otimes I_{N_p+1}\bigg) + \bigoplus_{i=1}^NH^i, \\ \nonumber
	&\text{(ii)}: \quad \nabla\eta(z_t) = (\nabla^2\eta)z_t - u_t,  
	\end{align}
	where $\bigoplus$ denotes the Kronecker sum of  matrices~\cite{horn2012matrix}, diag$(a)$ denotes the diagonal matrix of the vector a, $\mathbf{1}_N \in \mathbb{R}^N$ is the column vector of all ones, and the vector $u_t \in \mathbb{R}^{(N_p+1)N}$ is,
	\begin{align} 
		u_t &= [(u_t^1)^T,\dots,(u_t^N)^T]^T \ \text{with}, \\  u_t^i &= [d^i_{t-1|t-1}+s_{t-1},s_t,\dots,s_{t+\predHorzLen}]^T.
	\end{align}
\end{prop}
We have dropped the dependence of $z_t$ on the Hessian, as the Hessian is a constant matrix, where additionally, based on the form given in Proposition~\ref{prop:hessStruc}, it is symmetric, i.e., $\nabla^2\eta = (\nabla^2\eta)^T$ and positive definite.
\begin{prop} \label{prop:hessDiffBound}
Let $\bar{\zeta}^i = \zeta^i$ for all $i \in \{1, \dots, N \}$, then
\begin{align} \nonumber
\big\|P\nabla^2\eta - \nabla^2\eta P\big\| =0.
\end{align}
\end{prop}

%We include the above results as the above bound will be required in the stability analysis. Of particular interest is that the bound does not depend on $N$.

\begin{lem}[Theorem 2.1,~\cite{Daniel:1973}] \label{lem:stabQuad2}
	For any $s,\tau \in \mathbb{N}$, the following bound holds,
	\begin{align}
		\frac{1}{N}\|z^*_s-z^*_{\tau}\| \leq \frac{\bar{u}^*_{s,\tau}}{\lambda_{\text{min}}(\nabla^2\eta)},
	\end{align}
	where $\bar{u}^*_{s,\tau} = \|u_s^*-u^*_{\tau}\|$.
\end{lem}
\begin{proof}
	See~\cite{Daniel:1973}.
\end{proof} 

\begin{lem} \label{lem:invModel}
	For all $t\in\mathbb{N}$ the following holds,
	\begin{align}	\nonumber
		\frac{1}{N}\|Pz^*_{t-1} - z^*_{t}\| &\leq  \frac{\bar{g}^*_t}{\lambda_{\text{min}}(\nabla^2\eta)},
	\end{align}
	where $\bar{g}^*_t = \bar{u}^*_{t,t-1} + 2\tilde{u}^{*}_t$, $\tilde{u}^{*}_t = \|u^*_{t-1} - u^{*,0}_t\|$ and $u^{*,0}_t$ is the value that produces an optimal solution of all zeros.
\end{lem}
\begin{proof}
	See appendix.
\end{proof} 
This result will render itself useful for the stability analysis. Also necessary in our stability results is the class of $\mathcal{K}$ and $\mathcal{KL}$ functions, that hold their usual definitions as seen, e.g. in~\cite{khalil2002nonlinear}.
\ifx 0
\begin{lem} \label{lem:stabQuad}
	~\cite{Daniel:1973} Suppose $D^{*}_1$ and $D^*_2$ are the solutions to~\eqref{prob:primalOnly} with $S_t = S_1$ and $S_t = S_2$, respectively. Then letting,
	\begin{align} \nonumber
	\varepsilon = \left|S_1-S_2\right|, \quad \text{and} \quad \zeta^{\text{min}} = \min_{1\leq i\leq N}\zeta^i,
	\end{align}
	we have the following,
	\begin{align} \nonumber
	\left|\left|D^*_2-D^*_1\right|\right| \leq \frac{N\varepsilon}{\zeta^{\text{min}}} < \infty
	\end{align}
\end{lem}
\begin{proof}
	See~\cite{Daniel:1973}.
\end{proof}
This is essentially a ``stability'' result for a quadratic program. As we will see, this result will allow us to claim an input to state (ISS) stability result for the algorithm~\eqref{eq:distPrimAlg}, where the input is, for some $\alpha \in \mathbb{N}$, the difference of the reference signal: $\left|S_t-S_{t+\alpha}\right|$. 
\fi

\ifx 0
\begin{comment}
	The outer product in Proposition~\eqref{prop:hessStruc} is a consequence of the coupling term in~\eqref{prob:primalOnly}. The matrix $Z$ is a consequence of the regularization terms at each prediction time in~\eqref{prob:primalOnly}. This idea generalizes, coupling terms will form outer products whereas regularizations from terms similar to $Z$.
\end{comment}
\fi
\subsection{Stability: Main result}
Our main theoretical results for our proposed algorithm~\eqref{eq:distPrimAlg} is summarized in  Theorem~\ref{thm:thmPrimal}. If we treat the value $\|z_t - z^*_{t}\|$ as the ``state'' and an upper bound on the time varying aspects to the optimization problem as the ``input'', then Theorem~\ref{thm:thmPrimal} is a global input to state stability (ISS) result. 

%In addition, if time variations eventually vanish then one can achieve global asymptotic stability (GAS) to the asymptotic optimal solution, which is the result of Theorem~\ref{thm:GAS}. 

Practically, we want the magnitude $\|z_t - z^*_{t}\|$ to be small, as the optimal solution $z^*_t$ represents the value that optimally satisfies all of the specified criteria. %This conceptually must always be true, otherwise the problem is ill-posed.

The theorem below requires the following boundedness assumptions:
\begin{enumerate}
	\item[\textbf{A1}:] for all $t \in \mathbb{N}$, $\bar{g}^*_t < \bar{g} < \infty$, 
	\item[\textbf{A2}:] for all $t \in \mathbb{N}$, $\ell < t$, $\|Pu_{t-\ell} - u^*_t\| < \Delta < \infty$.
\end{enumerate}	
Then we denote $\bar{u} \triangleq \big(\frac{N\bar{g}}{\alpha\lambda_{\text{min}}(\nabla^2\eta)} + \Delta\big)$.
\begin{theorem}[Global-ISS] \label{thm:thmPrimal}
	 If assumptions \textbf{A1} and \textbf{A2} are satisfied, the step size $\alpha$ satisfies,
	\begin{align} \nonumber
			&\alpha \in \bigg(0, \frac{1}{\zeta^{\text{max}} + N}\bigg), \quad \text{where} \quad
			\zeta^{\text{max}} = \max_{1\leq i\leq N} \zeta^i,
	\end{align}
	and $\bar{\zeta}^i = \zeta^i$ for all $i\in\{1,\dots,N\}$,
	then for all $z_0 \in \mathbb{R}^{(N_p+1)N}$ there exists a $\Gamma \in \mathcal{K}$ and an $\Omega \in \mathcal{K}\mathcal{L}$ such that 
	\begin{align}\nonumber
	\|z_t - z^*_{t}\| \leq \Omega(\|z_0 - z^*_{0}\|,t) + \Gamma(\bar{u}) 
	\end{align}
	where $z_0$ is the initial iterate of~\eqref{eq:fullAlg}.
\end{theorem}
\begin{proof}
	See appendix.
\end{proof}
\ifx 0 
\begin{theorem}[GAS] \label{thm:GAS}
	In addition to the requirements of Theorem~\ref{thm:thmPrimal}, if
	\begin{align} \nonumber
		\lim_{t\rightarrow\infty}\bar{u}^*_t = 0, \ \text{and} \ \lim_{t\rightarrow\infty}u_t = u_{\infty},
	\end{align}
	then for all $z_0 \in \mathbb{R}^{(N_p+1)N}$
	\begin{align} \nonumber
		\lim_{t\rightarrow\infty}\|z_t-z^*_\infty\| = 0,
	\end{align}
	where $z^*_\infty$ is the optimal solution of~\eqref{prob:predMemResAlloc} corresponding to the input data $u_{\infty}$.
\end{theorem}
\begin{proof}
	See appendix.
\end{proof} 
\fi
In Theorem~\ref{thm:thmPrimal} we have developed conditions on the stepsize in terms of the readily available problem data that will give a stability result for time varying reference signals. 

%In Theorem~\ref{thm:GAS} we show that if the time varying nature eventually ``vanishes'' an asymptotic result is obtained, under the same stepsize conditions as Theorem~\ref{thm:thmPrimal}.
%Here we have identified the ``input'' to the system as the difference of the signal $S_t$ and the state as the difference from the current iterate to the optimal point. Clearly, if if the input is zero, then the limit of $\left|\left|D_t-D^*\right|\right|^2$ is zero as $\Omega(\cdot,\cdot)$ is, by definition, decreasing in its second argument. So we only provide a proof for the ISS result, as the 0-GAS result follows trivially. 

\ifx 0
\begin{comment}
We see that the bound on $\alpha$ is not effected by allowing the individual agents to make predictions. In some sense, this means that the stability of the system is not effected by individual agents predictions. Rather, stability is dominated by any coupling terms that may appear in the objective. These coupling terms will form outer products in the Hessian (as observed in Proposition~\ref{prop:hessStruc}), and require that the step size be bounded by roughly one over the number of agents that contribute to the coupling term.

However, the predictions help give a higher rate of convergence.
\end{comment}
\fi
\ifx 0
\subsection{Practical Convergence Estimate Rate}
The quantity $M=\left|\left|I - \alpha I\otimes Z-\alpha(I \otimes\mathbf{1}^T)\right|\right|$ plays the defining role in how fast the algorithm converges. This is inevitably related to the singular values of the Hermitian matrix $\mathcal{M} = I - \alpha I\otimes Z-\alpha(I \otimes\mathbf{1}^T)$. Since $\mathcal{M}$ is Hermitian, we have the \emph{Weyl inequality} for the matrix $\mathcal{M}$,
\begin{align} \nonumber
	\sum_{i=1}^{N}\sigma_i(\mathcal{M}) &\leq \sum_{i=1}^{N}\sigma_i(I-\alpha I\otimes Z) + \sum_{i=1}^{N}\sigma_i(-\alpha(I \otimes\mathbf{1}^T)) \\
	&=\sum_{i=1}^{N}(1-\alpha\zeta^i) - \sum_{i=1}^{N}\alpha N,
\end{align}
where $\sigma_i(A)$ is the $i^{th}$ singular value of the matrix $A$. 
\begin{align}
	M = \left|\sigma_{\text{max}}(\mathcal{M})\right| \leq \left|\sum_{i=1}^{N}\big(1-\alpha\zeta^i - \alpha N\big)\right|.
\end{align}
Further suppose that $\alpha = 1/N$ so that,
\begin{align}
	M \leq \left|\frac{\sum_{i=1}^{N}\zeta^i}{N}\right|
\end{align}
\fi

\ifx 0
\subsection{High gain stability margin}
As a remark, we observe that the bound on alpha given is tight if $\zeta^i = \zeta^0$ for all $i \in \{1,\dots,N\}$. That is, in the homogeneous weighting scenario, no other higher bound exists for the given algorithm. The reason for this is a stronger form of the Gershgorin Circle Theorem, the maximum eigenvalue can only be on the edge of a Gershgorin disk if and only if the edge of all disks coincide. However, this will \emph{not} be satisfied with heterogeneous weights. so that the bound obtained on $\alpha$ from the circle theorem is strict.

While this facet seems purely mathematical, this result is actually quite practical. If the central authority has inaccurate knowledge of $N$ (the number of loads), we have baked in 2 factors of safety: (i) heterogeneity and (ii) sufficiency in the bounds on $\alpha$ based on the circle theorem. Recall, the stability results of Theorem~\ref{thm:thmPrimal} require $\left|\lambda_{\text{max}}\big(I-\alpha\nabla^2\eta\big)\right| < 1$, and we have elected $\alpha$ so that 
\begin{align}
	\alpha < \frac{1}{\zeta^{\text{max}}+N} <  \frac{1}{\lambda_{\text{max}}\big(\nabla^2\eta\big)} < \frac{2}{\lambda_{\text{max}}\big(\nabla^2\eta\big)} = \alpha_{\text{max}},
\end{align}
where $\alpha_{\text{max}}$ is the maximum allowable $\alpha$. 
This motivates a natural definition for the stability margin of the system.
\begin{definition}
	The high gain stability margin $\rho\in (0,\frac{1}{2}]$ is the value such that the following occurs,
	\begin{align} \nonumber
	 \frac{2}{\lambda_{\text{max}}\big(\nabla^2\eta\big)} - \frac{1}{\zeta^{\text{max}}+\rho N} = 0.
	\end{align}
\end{definition} 
Interestingly, this means that the value $N$ can be put to $N/2 + \epsilon$ for all $\epsilon > 0$ and the algorithm is still stable in the sense of Theorem~\ref{thm:thmPrimal}. Meaning, whatever individual is responsible for knowing $N$ in order to implement the algorithm, this individual does not need a good estimate of $N$ to preserve stability. 
\fi

\section{Numerical Examples} \label{sec:numExp}
Here we offer numerical examples to validate the result from Theorem~\ref{thm:thmPrimal}. This involves simulating the algorithm~\eqref{eq:distPrimAlg} on various types of data. We provide two scenarios for this: Scenario 1 (S1) a step reference that makes problem~\eqref{prob:compProb2} not feasible so to illustrate the integrator windup of the dual ascent method and Scenario 2 (S2) our proposed method tracking Bonneville Power Administrations (BPA) balancing reserves deployed (BRD) signal to illustrate the effectiveness of our algorithm tracking a time varying signal. 

In both scenarios: (i) each load is given a set of parameter values obtained by a linear spacing between the maximum and minimum values found (along with the other relevant simulation parameters) in Table~\ref{tab:simParam} and (ii) the sampling time is $T_s = 5$ minutes.
%In any case three important conceptual metrics are: (i) How well can the loads track the given reference signal (ii) How well can the loads satisfy their QoS while tracking this signal and (iii) is the consumers privacy protected?

%The point of these experiments will be that our proposed method is the only method that can achieve simultaneously the above three things. The comparison scenario 1 and 2 will fail to satisfy (ii) above, as by design (not our design!) they cannot account for the loads QoS. Additionally, comparison scenario 2 either fails (i) or (iii); it requires knowledge of \emph{every} loads QoS to ensure that the problem is feasible or integrator windup will occur and cause poor tracking performance.

%In all of the scenarios, data for the reference signal is obtained from Bonneville Power Administration (BPA) a balancing authority in the pacific northwest. The other relevant simulation parameters are given in Table~\ref{tab:simParam}.

\begin{table}[t]
	\centering
	\caption{Simulation Parameters}
	\label{tab:simParam}
	\begin{tabular}{|| l c c|| c c c||}
		Par. & Unit & value & Par. & Unit & value \\ 
		N & hundred & 1 & $\alpha$ & N/A & $\frac{0.99}{\zeta^{\text{max}}+N}$  \\ 
		$\zeta^{\text{min}},\zeta^{\text{max}}$ & N/A & 0.1, 4 & 
		$e^{\text{min}}_L,e^{\text{max}}_H$ & kWh & 0, 4 \\ 
		$d^{\text{min}}_L,d^{\text{max}}_H$ & kW & 0, 10 & $r^{\text{min}}_L,r^{\text{max}}_H$ & kW & -0.50, 0.50 \\
		%$\eta$ & $N/A$ & 2.5 \\
		%$T_a$ & $^{\circ}C$ & 30 \\
		%$\theta_{set}$ & $^{\circ}C$ & 21 \\
		%$\Delta$ & $^{\circ}C$ & 2 \\
		%$T_s$ & Mins. & 2 \\
	\end{tabular}
\end{table}

\subsection{Scenario 1: Integrator Windup of dual ascent}
The first example we illustrate is the ``integrator windup'' behavior that the dual algorithm suffers when problem~\eqref{prob:compProb2} is not feasible, as described in Section~\ref{sec:dualMethod}. The result of this is shown in Figure~\ref{fig:ingWindDual}. When the resource allocation problem~\eqref{prob:compProb2} is not feasible, the dual variable update~\eqref{eq:unstableInt} will continue to integrate non-zero area. It then takes dual ascent time to reach zero steady state error once feasibility is regained. It is worth noting that the two regions of integrated area in Figure~\ref{fig:ingWindDual} are equivalent.  

For comparison we also utilize our proposed algorithm with solely the magnitude constraints~\eqref{eq:powerCon} and $N_p = 0$, which does not suffer from integrator windup. 

%We also provide results for a time varying reference, shown in Figure~\ref{fig:ingWindTV}. Observing that at hours 5, 10, and 22 we see the Dual method has trouble tracking the reference as the dual variable was allowed to grow while the capacity was violated.

\begin{figure}[h]
	\centering
	\includegraphics[width=1\columnwidth]{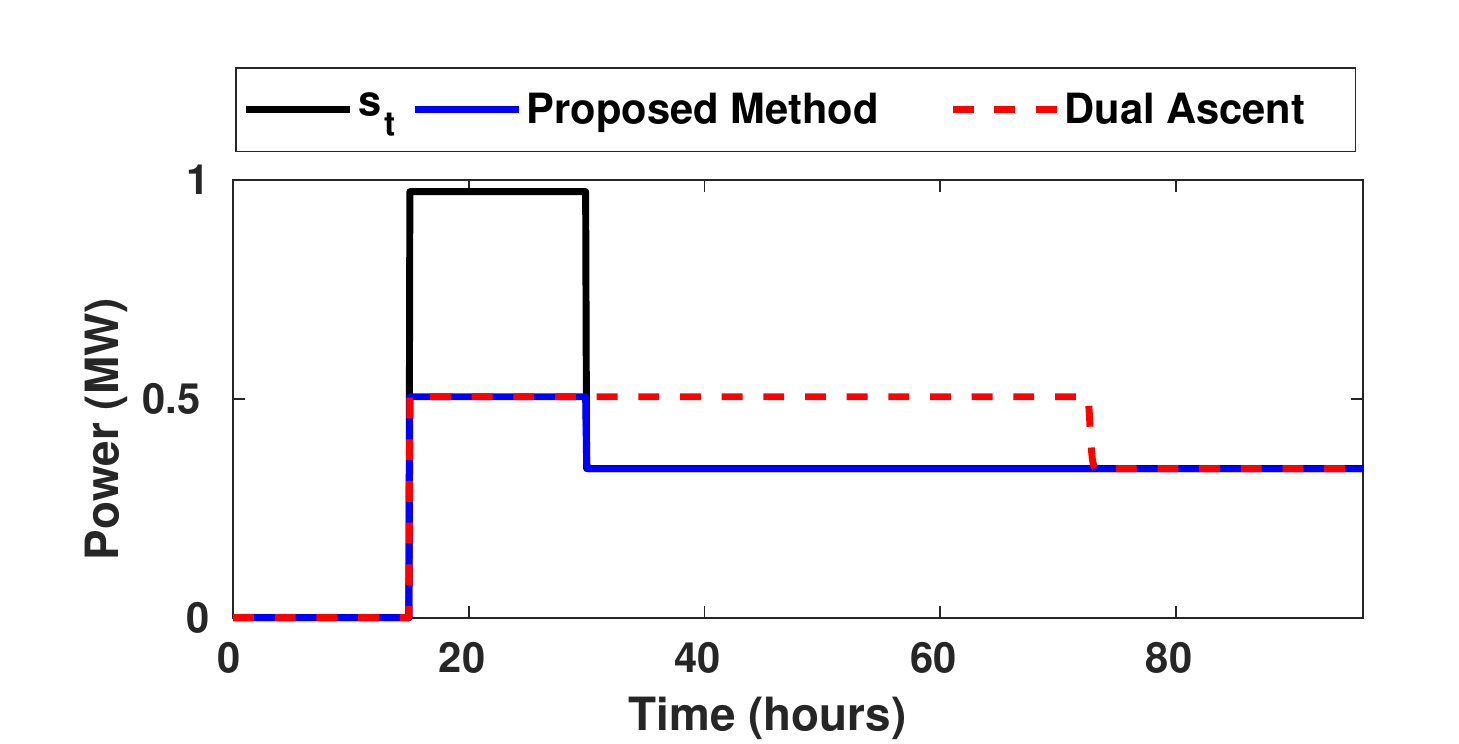}
	\caption{Integrator Windup of dual ascent with step response reference.}
	\label{fig:ingWindDual}
\end{figure}
\ifx 0
\begin{figure}[h]
	\centering
	\includegraphics[width=1\columnwidth]{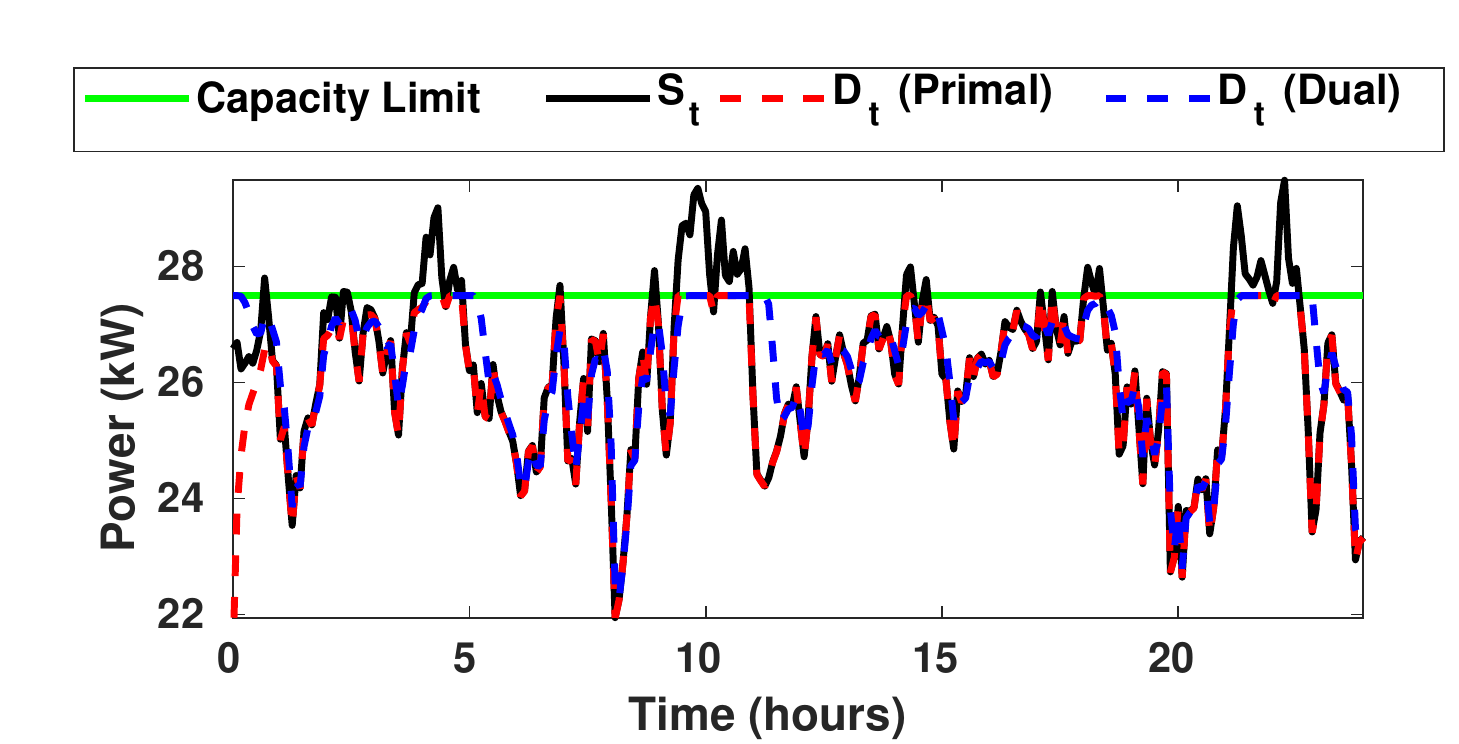}
	\caption{Integrator Windup of Dual Method with time varying reference.}
	\label{fig:ingWindTV}
\end{figure}
\fi

\ifx 0
\subsection{Robustness to Weight Heterogeneity}
We illustrate here the robustness of the methods to heterogeneity in the weights $\zeta^i$. We find that when the tracking a step reference the dual method is more robust to weight design. However, when tracking a time varying reference the Primal method is more robust to weight heterogeneity. The results for the time varying and step reference signals are shown in Figure~\ref{fig:timeRefHetWeight} and Figure~\ref{fig:stepRefHetWeight}, respectively.

As the dual method acts as an integrator, the results for the step reference are not too surprising. Since some of the weights are higher than the weight on reference tracking for the primal method, the result is not too surprising for the step reference for the primal method. It should be noted that the primal method has reached the optimal solution, however the optimal solution is not perfect reference tracking. We see that the primal method significantly out performs the dual method with the time varying reference. 

\subsection{General Remarks}
We have compared the primal and the dual algorithm for two different cases. From the results we can conclude that the dual algorithm is preferred for loads tracking step reference signals. Whereas the primal algorithm is preferred for loads tracking time varying reference signals. Furthermore, irrespective of the reference signal the primal method is more robust when it comes to reference signals that violate the capacity of the collection. The dual method exhibits integrator windup when the reference signal violates the capacity of the collection, whereas the primal method does not.  
\fi
\ifx 0
\begin{figure}[h]
	\centering
	\includegraphics[width=1\columnwidth]{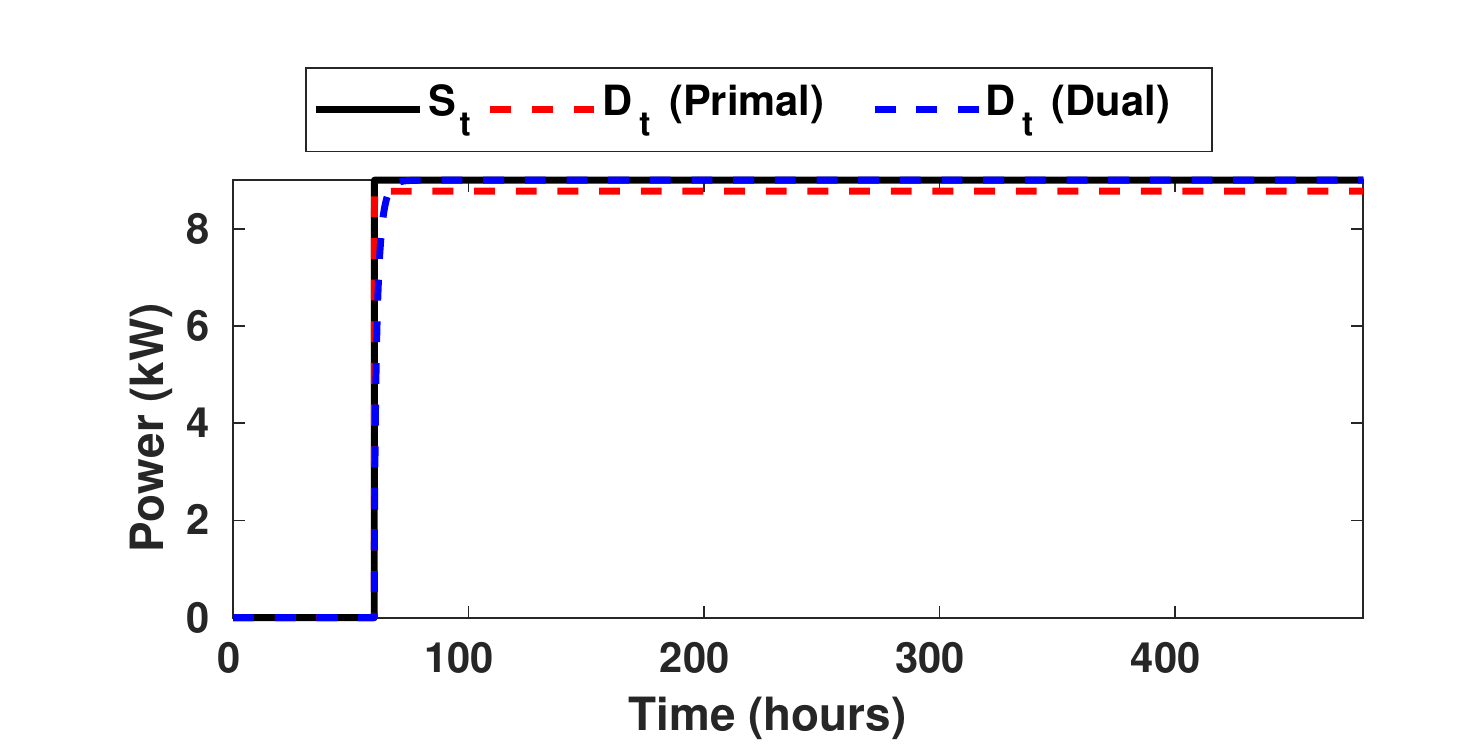}
	\caption{Tracking the step reference with Heterogeneous Weights.}
	\label{fig:stepRefHetWeight}
\end{figure}
\begin{figure}[h]
	\centering
	\includegraphics[width=1\columnwidth]{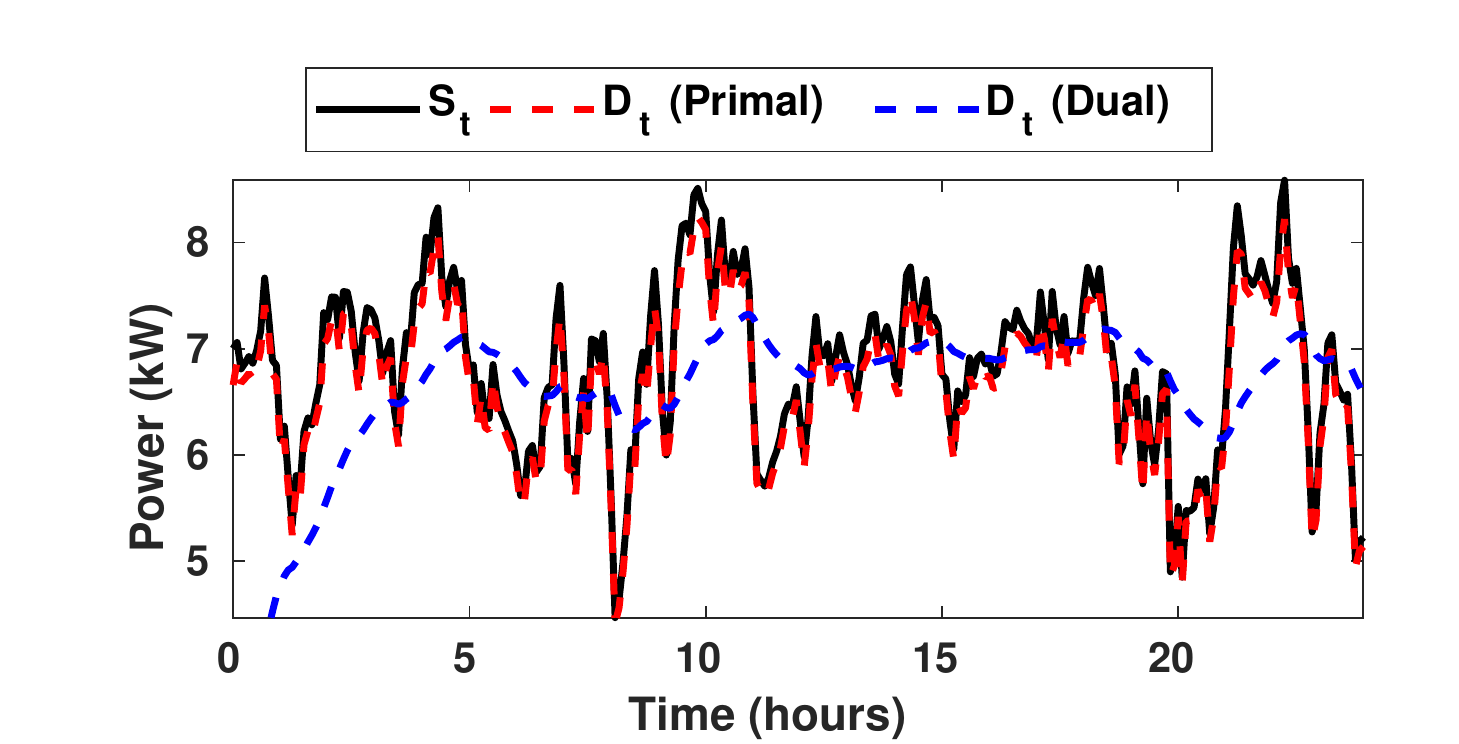}
	\caption{Tracking the time varying reference with Heterogeneous Weights.}
	\label{fig:timeRefHetWeight}
\end{figure}
\fi

\subsection{Scenario 2: Tracking BPA's BRD}
With our proposed method, we track a time-varying reference with a prediction horizon of $N_p = 5$; see Figure~\ref{fig:refTrackPropMethod}. Since the data obtained from BPA is on the order of GW, we scale the reference down to satisfy the magnitude constraint. However, this is \emph{not} required for the success of the algorithm, only to aid in exposition of the results. 

The 1-norm tracking error of the signal in Figure~\ref{fig:refTrackPropMethod} is $16.3 \%$, and can be attributed to 2 factors: (i) the reference is only guaranteed to satisfy the magnitude constraint~\eqref{eq:powerCon} so it may not be feasible for the other constraints and (ii) the algorithm only guarantees ISS and not asymptotic tracking. However, from experience we believe (i) to be the contributing factor. Other numerical experiments conducted suggest that it is possible to make the error quite small by increasing $N_p$ if the constraints are all feasible.

\begin{figure}[h]
	\centering
	\includegraphics[width=1\columnwidth]{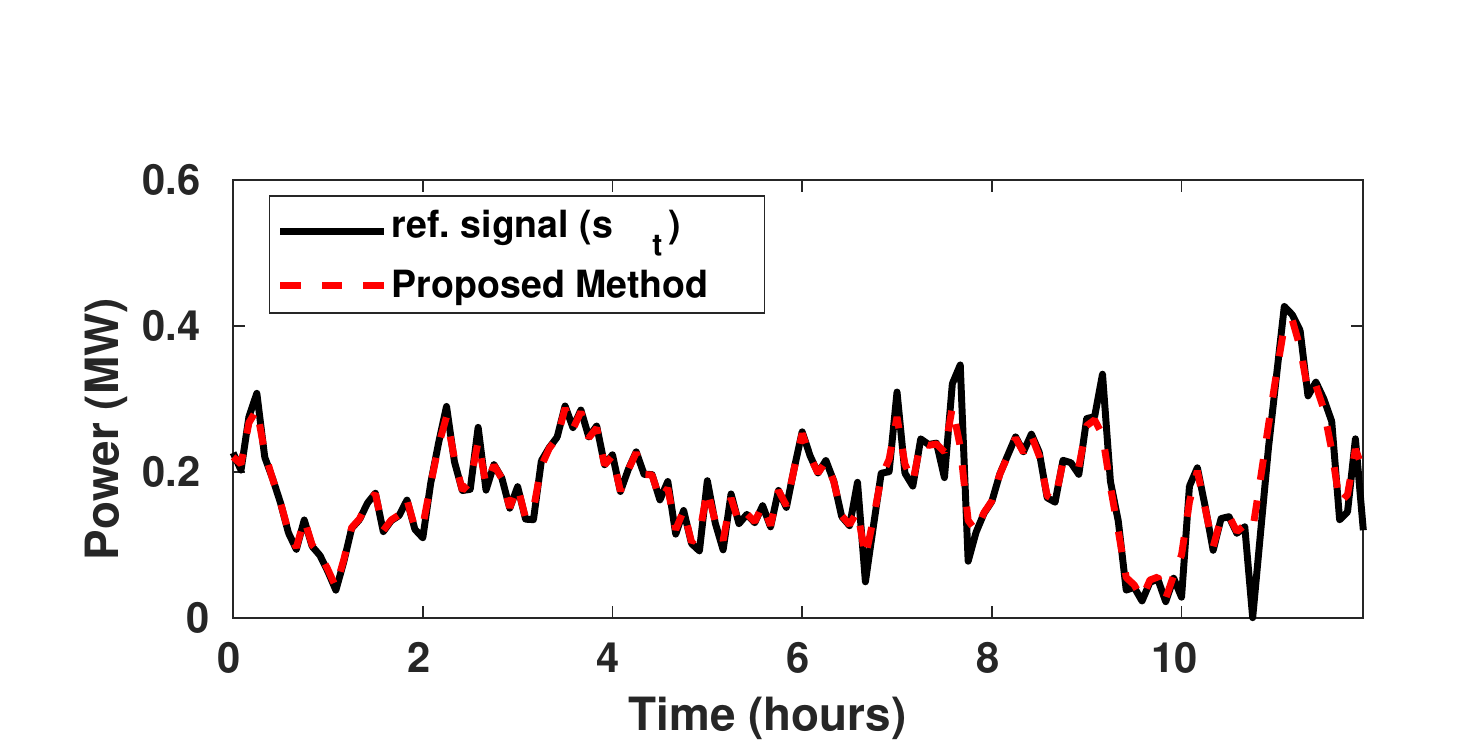}
	\caption{Tracking the time varying reference with the proposed method.}
	\label{fig:refTrackPropMethod}
\end{figure}

\section{Conclusion} \label{sec:conc}
We propose a real time optimization algorithm with distributed computation and hierarchical communication structure for the resource allocation of flexible loads in the smart grid. Our algorithm has two key innovations: (i) the utilization of predictions and (ii) a state augmentation technique to handle dynamic constraints.

Future work includes: (i) analyzing further the effects of the state augmentation technique, similar to the penalty method technique applied in~\cite{CherCortes:2018TAC} and (ii) the development of asymptotic results for constrained time varying optimization.    

\ifshowArxivAlt
\bibliographystyle{IEEEtran}
\bibliography{/home/austin/allBibFiles/Barooah,/home/austin/allBibFiles/optimization,/home/austin/allBibFiles/ControlTheory,/home/austin/allBibFiles/basics,/home/austin/allBibFiles/distributed_control,/home/austin/allBibFiles/grid}
%\bibliography{litReview,../../bibFiles/dicelab_bib/optimization,../../bibFiles/dicelab_bib/Barooah,../../bibFiles/dicelab_bib/ControlTheory,../../bibFiles/dicelab_bib/Basics,../../bibFiles/dicelab_bib/distributed_control,../../bibFiles/dicelab_bib/grid}
\fi

\ifshowArxiv

\fi

\section*{Appendix}
\ifx 0
\subsection{Proof of Lemma~\ref{lem:invModel}}
%\begin{proof}
	The inventory system is developed under the assumption of perfect predictions, so that $z^*_{t-1}[j]=z^*_t[j-1]$ for $i \in \{2,\dots,N_p+1\}$. For $i=1$, we must add and subtract the correct values to make the inventory recursion correct.
	
	To prove the bound, consider $\|Pz^*_{t-1}-z^*_t\|$:
	\begin{align} \nonumber
		&\leq \|B\|\|P\|\|z^*_{t-1}-z^*_t\| \\ \nonumber
		&= \|z^*_{t-1}-z^*_t\| = \|\Pi_{\mathcal{D}}(\psi^*_{t-1})-\Pi_{\mathcal{D}}(\psi^*_{t})\| \\ \nonumber
		&\leq M(\alpha)\|z^*_{t-1}-z^*_t\| + \alpha\|u^*_t-u^*_{t-1}\|, \\ \nonumber
		&\leq \frac{\alpha \|u^*_t-u^*_{t-1}\|}{1-M(\alpha)} =\frac{\alpha \bar{u}^*_t}{1-M(\alpha)}
	\end{align}
%\end{proof}
\fi
\subsection{Proof of Lemma~\ref{lem:invModel}}
Proceeding directly, by the triangle inequality we have
\begin{align} \nonumber
	\frac{1}{N}\|Pz^*_{t-1} - z^*_t\| &\leq \frac{1}{N}\|z^*_{t-1}-z^*_{t}\| + \frac{1}{N}\|Pz^*_{t-1} - z^*_{t-1}\|, \\
	&\leq \frac{\bar{u}^*_{t,t-1}}{\lambda_{\text{min}}(\nabla^2
		\eta)} + \frac{2}{N}\|z^*_{t-1}\|.
\end{align}
We can bound $\|z^*_{t-1}\|$ by using Lemma~\ref{lem:stabQuad2} where $z^*_t$ will be zero when $u^*_t = u^{*,0}_t$, yielding
\begin{align} \nonumber
	\frac{1}{N}\|Pz^*_{t-1} - z^*_t\| &\leq \frac{\bar{u}^*_{t,t-1}}{\lambda_{\text{min}}(\nabla^2
		\eta)} + \frac{2\tilde{u}^{*}_{t}}{\lambda_{\text{min}}(\nabla^2
		\eta)} = \frac{\bar{g}^*_t}{\lambda_{\text{min}}(\nabla^2
		\eta)},
\end{align}
which is the desired result.
\subsection{Proof of Theorem~\ref{thm:thmPrimal}}
%\begin{proof}
	We start with the developed vectorized notation, 
	\begin{align} \nonumber
	\|z_t-z^*_t\|&= \|\Pi_{\mathcal{D}}(P\psi_{t-1})-\Pi_{\mathcal{D}}(\psi^*_{t})\| \leq\|P\psi_{t-1}-\psi^*_{t}\|  \\ \nonumber
	&=\|Pz_{t-1} - z^*_t- \alpha\big(P\nabla\eta(z_{t-1})-\nabla\eta(
	z^*_{t})\big)\|, 
	\end{align}
	where the inequality is by the non-expansive property of the projection operator. Working with the gradient terms we have from Proposition~\ref{prop:hessStruc} that,
	\begin{align} \nonumber
	P(\nabla\eta(z_{t-1})) = P\bigg((\nabla^2\eta)z_{t-1} - u_{t-1}\bigg)
	\end{align} 
	so that $P\nabla\eta(z_{t-1})-\nabla\eta(z^*_{t}):$
	\begin{align} \nonumber
	&= \nabla^2\eta\bigg(Pz_{t-1} - z^*_{t}\bigg) + \bigg(u^*_t - Pu_{t-1} \bigg) \\ \nonumber &+ \bigg(P(\nabla^2\eta) - (\nabla^2\eta)P\bigg)z_{t-1}.  
	\end{align}
	Substituting this result into to the original quantity of interest and applying the triangle inequality, we have:
	\begin{align} \nonumber
	\|z_t-z^*_t\| &\leq M(\alpha)\|Pz_{t-1}-z^*_t\|
	+\alpha\|(Pu_{t-1}-u^*_t)\|\\ \nonumber&+\alpha\bigg\|\bigg(P(\nabla^2\eta) - (\nabla^2\eta)P\bigg)\bigg\|\|z_{t-1}\|, \\ \nonumber
	&\leq M(\alpha)\|Pz_{t-1}-z^*_t\|
	+\alpha\|(Pu_{t-1}-u^*_t)\|,
	\end{align} 
	where $M(\alpha) = \|I - \alpha\nabla^2\eta\|$. The third term above is eliminated from our choice of $\bar{\zeta}^i=\zeta^i$. As the Hessian is positive definite, it is possible to pick an $\alpha$ so that $M(\alpha)<1$. We take this fact for granted now, and later in the proof provide the bound found in the theorem. Now utilizing the triangle inequality we have that,
	\begin{align} \nonumber
	\|z_{t}-z^*_t\|&\leq M(\alpha)\bigg(\|z_{t-1}-z^*_{t-1}\|+\|Pz^*_{t-1}-z^*_t\|\bigg)
	\\ \nonumber
	&+\alpha\|(Pu_{t-1}-u^*_t)\|.
	\end{align} 
	Now applying the results of Lemma~\ref{lem:invModel} we  have,
	\begin{align} \nonumber
	\|z_{t}-z^*_t\|&\leq M(\alpha)\bigg(\|z_{t-1}-z^*_{t-1}\|+\frac{N\bar{g}^*_t}{\lambda_{\text{min}}(\nabla^2
		\eta)}\bigg)
	\\ \nonumber
	&+\alpha\|(Pu_{t-1}-u^*_t)\|.
	\end{align}
	We iterate this backwards a total of $t$ times to reach $t=0$, yielding:
	\begin{align} \nonumber
	\|z_{t}-z^*_t\| &\leq M^{t}(\alpha)\|z_{0}-z_0^*\|  \\ \nonumber
	&+\alpha\sum_{\ell=1}^{t}M^{t-\ell}(\alpha)\bigg(\|P u_{t-\ell}-u^*_t\| + \frac{N\bar{g}^*_t}{\alpha\lambda_{\text{min}}(\nabla^2
		\eta)}\bigg).
	\end{align}
	Now, from our assumptions we can bound the quantity in parentheses in the summation by $\bar{u}$ yielding,
	\begin{align} \nonumber
	\|z_t - z^*_t\| &\leq M^{t}(\alpha)\|z_{0}-z_0^*\| +\frac{\alpha\bar{u}}{1-M(\alpha)} , \\ \nonumber
	&\leq \Omega(\|z_{0}-z_0^*\|,t) + \Gamma(\bar{u}),
	\end{align}
	where, as desired, it can be easily verified that $\Omega \in \mathcal{KL}$ and $\Gamma \in \mathcal{K}$ as long as $M(\alpha) < 1$, which we ensure next.
	
	Now that the ISS result has been obtained, we show how the range on $\alpha$ is obtained to guarantee $M(\alpha)=\|I - \alpha\nabla^2\eta\| < 1$. By definition we have, $\|I - \alpha\nabla^2\eta\|$:
	\begin{align} \nonumber
	 &=\max\big\{\left|\lambda_{\text{min}}\big(I - \alpha\nabla^2\eta\big)\right|,\left|\lambda_{\text{max}}\big(I - \alpha\nabla^2\eta\big)\right|\big\},
	\end{align}
	since $I - \alpha\nabla^2\eta$ is symmetric and where $\lambda_{\text{max}}(A)$ and $\lambda_{\text{min}}(A)$ are the maximum and minimum eigenvalue of the matrix $A$, respectively. If we denote $\lambda_i(\nabla^2\eta)$ the $i^{th}$ eigenvalue of $\nabla^2\eta$, then $\lambda_i(I-\alpha\nabla^2\eta)$, the $i^{th}$ eigenvalue of $I-\alpha\nabla^2\eta$, is
	\begin{align} \nonumber
		\lambda_i(I-\alpha\nabla^2\eta) = 1 -\alpha\lambda_i(\nabla^2\eta),
	\end{align}
	which is obtained by considering the eigendecomposition of $\nabla^2\eta$.
	Thus, we seek to guarantee $M(\alpha) < 1$, and it is sufficient to require 
	\begin{align} \nonumber
	0 < \alpha\lambda_{\text{max}}\big(\nabla^2\eta\big) < 1,
	\end{align}
	which immediately leads to the lower bound $\alpha>0$, since the Hessian is positive definite. To obtain the upper bound we apply the Gershgorin circle theorem~\cite{varga2010gervsgorin}. This is readily applicable based on the structure of the Hessian found in Proposition~\ref{prop:hessStruc}. This yields the following sufficient lower and upper bound on $\alpha$ for $M(\alpha) < 1$, 
	\begin{align} \nonumber
	&\alpha \in \bigg(0, \frac{1}{\zeta^{\text{max}} + N}\bigg).
	\end{align} 
%	\end{proof}

\ifx 0
\subsection{Proof of Theorem~\ref{thm:GAS}}
To show the result, we will utilize the results of Theorem~\ref{thm:thmPrimal} and Lemma~\ref{lem:stabQuad2}. From the former we have that,
\begin{align} \nonumber
	0 \leq \|z_t-z^*_{t}\| \leq \Omega(\|z_{0}-z_0^*\|,t) + \Gamma(\bar{u}).
\end{align}
From the latter we have that,
\begin{align} \nonumber
	0 \leq \|z^*_{t}-z^*_{\infty}\| \leq \frac{N\|u^*_{t}-u^*_\infty\|}{\lambda_{\text{min}}(\nabla^2\eta)}.
\end{align}
Now combining, applying the triangle inequality, and taking the limit we obtain the desired result, 
\begin{align} \nonumber
	0 \leq \lim\limits_{t\rightarrow\infty}\|z_t-z^*_\infty\| \leq 0.
\end{align}
The limit evaluates to zero based on the assumptions given in Theorem~\ref{thm:GAS} and the properties of $\mathcal{KL}$ functions.
\fi
\end{document}